\documentclass[aps,prd,twocolumn,superscriptaddress,nofootinbib,longbibliography]{revtex4-2}

\usepackage[T1]{fontenc}
\usepackage[utf8]{inputenc}
\usepackage{microtype}
\usepackage{amsmath,amssymb,mathtools,bm,amsthm}
\usepackage{graphicx}
\usepackage{booktabs}
\usepackage{xcolor}

\usepackage[colorlinks=true, urlcolor=blue, linkcolor=blue, citecolor=red, hyperindex=true, linktocpage=true]{hyperref}

\newcommand{\Tr}{\operatorname{Tr}}
\newcommand{\dist}{\operatorname{dist}}
\newcommand{\norm}[1]{\left\lVert #1\right\rVert}
\newcommand{\abs}[1]{\left|#1\right|}
\newcommand{\one}{\mathbf 1}
\newcommand{\eps}{\varepsilon}
\newcommand{\E}{\mathbb E}
\newcommand{\diam}{\operatorname{diam}}
\newcommand{\cent}[1]{#1^{\circ}}

\newtheorem{theorem}{Theorem}[section]
\newtheorem{lemma}[theorem]{Lemma}

\theoremstyle{remark}

\begin{document}

\title{From Lyapunov Growth to Weak Quantum Chaos}

\author{Sa\v{s}o Grozdanov}
\affiliation{Faculty of Mathematics and Physics, University of Ljubljana, Jadranska 19, SI-1000 Ljubljana, Slovenia}
\affiliation{Higgs Centre for Theoretical Physics, University of Edinburgh, Edinburgh, EH8 9YL, Scotland, U.K.}

\author{Toma\v{z} Prosen}
\affiliation{Faculty of Mathematics and Physics, University of Ljubljana, Jadranska 19, SI-1000 Ljubljana, Slovenia}
\affiliation{Institute of Mathematics, Physics and Mechanics, Jadranska 19, SI-1000 Ljubljana, Slovenia}


\begin{abstract}
Local out-of-time-order correlators (OTOCs) measure the growth of initially small commutators. For bounded local operators, this growth ends when the local commutator becomes of order one, but in an extended system the region over which it is large can continue to expand. We study this spatial proliferation using three related observables: the local squared commutator $C_r(t)$, its spatial sum $\bar C(t) = \sum_r C_r(t)$, and the symmetrically regularised connected dOTOC $D(t)$ of two extensive observables. We show that Lieb--Robinson locality gives $\bar C(t) \lesssim t^d$. We then prove that, in a Gibbs state with uniform exponential clustering, the full dOTOC obeys the same bound, $D(t) \lesssim t^d$, improving our earlier suboptimal $t^{3d}$ bound. We illustrate these results in two quantum many-body systems: a Brownian coupled-cluster model and a finite-temperature SYK chain. In both cases, an exact disorder symmetry relates $D(t)$ to $\bar C(t)$ and  allows us to follow the same extensive observable from the pre-scrambling Lyapunov regime into the post-scrambling growth associated with an expanding
saturated region. This late-time power-law growth is the concrete
realisation of the weak quantum chaos introduced in our earlier work.
\end{abstract}

\maketitle

\section{Introduction}
\label{sec:introduction}

Chaos is the phenomenon whereby small differences in initial conditions make information about the initial state increasingly difficult to recover at later times. In recent years, out-of-time-ordered correlation functions (OTOCs) have emerged as a particularly useful local measure of this loss of information in quantum many-body systems, quantum field theories and the physics of black holes~\cite{SekinoSusskind2008,ShenkerStanford2014,Roberts:2014isa,Shenker:2014cwa,Maldacena:2015waa,Stanford:2015owe,Aleiner:2016eni,Nahum2018,vonKeyserlingk2018,Rakovszky:2017qit,XuSwingle2019,Swingle2018,Grozdanov:2017ajz,Grozdanov:2018atb,Craps:2019rbj,Craps:2020ahu,Jahnke:2018off,Xu:2022vko,Grozdanov:2018kkt}.  In their usual local form, they quantify the growth of an initially small perturbation and, in systems with a parametrically long scrambling regime, can exhibit exponential behaviour governed by a quantum Lyapunov exponent.

The goal of this work is to revisit the temporal growth of OTOCs before and after local scrambling.  We will distinguish a local OTOC, its spatial integral (sum), and the dOTOC (density of OTOC) of extensive operators introduced in Ref.~\cite{Kukuljan:2017xag}.  The distinction becomes important once a local commutator has become order one. Its local amplitude can then no longer grow, while the region over which it remains large may continue to increase.  The original motivation to introduce the dOTOC was precisely as a diagnostic of chaos (weak quantum chaos) in systems in which the local dimension of Hilbert space, or the rank of the gauge group in a quantum field theory, does not permit late-time growth of a local OTOC~\cite{Kukuljan:2017xag} (see also Refs.~\cite{Lin:2018tce,Keselman:2020fmo}).

Let us first discuss standard OTOCs. Consider local Hermitian operators $w_0$ and $v_r$, where the subscripts denote their spatial locations. We work in a thermal state $\rho=e^{-\beta H}/Z$, where $\beta$ is the inverse temperature.  Following Ref.~\cite{Maldacena:2015waa}, define the normalised out-of-time-ordered four-point function
\begin{equation}
F_r(t) \equiv 
\frac{\Tr\!\left(\rho^{1/4} v_r \rho^{1/4} w_0(t) \rho^{1/4} v_r \rho^{1/4} w_0(t)\right)}
{\Tr\!\left(\rho^{1/2} v_r \rho^{1/2} v_r\right)\Tr\!\left(\rho^{1/2} w_0 \rho^{1/2} w_0\right)}.
\label{eq:F-def}
\end{equation}
At zero separation, $r=0$, we write $F(t) =F_0(t)$.  For a large-$N$ theory, Maldacena, Shenker and Stanford (MSS) showed in Ref.~\cite{Maldacena:2015waa} that, in an exponential regime, the `four-point deficit' $1-F(t)$ behaves as
\begin{equation}
1-F(t)\sim \eps_N e^{\lambda_L t},
\qquad
\lambda_L\le \frac{2\pi}{\beta},
\qquad
\eps_N\ll1,
\label{eq:mss-window}
\end{equation}
when the assumptions specified in Ref.~\cite{Maldacena:2015waa} hold.  This gives the upper bound on the quantum Lyapunov exponent $\lambda_L$.  Throughout the paper, $N$ denotes the large internal parameter controlling a parametrically long scrambling regime. Depending on the model, it can be a gauge-group rank, a number of flavours, or the number of microscopic degrees of freedom in a cluster.

A closely related local quantity is the connected symmetrically regularised squared commutator
\begin{align}
C_r(t)
&\equiv\Tr\!\left[\rho^{1/2}\left(\cent{[w_0(t),v_r]}\right)^\dagger
\rho^{1/2} \cent{[w_0(t),v_r]}\right],
\nonumber  \\
\cent A&  =  A-\Tr(\rho A)\one.
\label{eq:C-def}
\end{align}
At finite temperature, it is useful to distinguish between $F_r$ and $C_r$.  The MSS argument applies directly to the analytic four-point function $F_r(z)$.\footnote{For important subtleties regarding the choice of the regularisation in the definition of  $F_r(t)$, see Ref.~\cite{Romero-Bermudez:2019vej}.} The quantity $C_r(t)$ is the squared commutator and is more directly suited to following the spatial spreading of the initial disturbance. It is also the form of the OTOC used in the early studies of quantum chaos~\cite{Larkin:1969,Kitaev:2015holography}. The distinction can be important dynamically.  In particular, the symmetric commutator need not have overlap with the fastest mode present in $F_r$.  In Sachdev--Ye--Kitaev (SYK)-type  systems, for example,  the maximally chaotic contribution can cancel from the commutator square (see Ref.~\cite{BlakeLiu2021}).

Importantly, in theories with finite $N$, the exponential regime must end when the local commutator becomes order one. If $t_*(N)$ is defined parametrically by $\eps_N e^{\lambda_L  t_*}\sim1$, then $t_*\to\infty$ as $N\to\infty$. Taking $N=\infty$ first at fixed time removes the saturation regime and leaves only the pre-scrambling expansion. In this work, we will be interested in the large but finite $N$ evolution through $t_*(N)$, or,  equivalently,  in a correlated large-$N$ limit in which $t$ scales with $t_*(N)$.

It is useful to keep in mind that local exponential growth alone does not establish spatially extensive scrambling. If the same local dynamics is spatially decoupled, the local OTOC at each site can be unchanged while a disturbance never reaches another site. Once the local commutator saturates, an extensive commutator per unit volume then saturates as well. In an interacting extended system, the region in which the local commutator is order one can keep growing after its amplitude has saturated.

To make this distinction precise, consider a periodic lattice $\Lambda$ of $L=|\Lambda|$ sites in $d$ spatial dimensions and the extensive operators
\begin{equation}
W = \sum_x w_x,  \qquad
V =  \sum_y v_y,
\label{eq:WV}
\end{equation}
where $w_x$ and $v_y$ are bounded Hermitian operators of fixed finite support. When the two sums are resolved into local pairs, we refer to $x$ as the source position and $y$ as the target position. We use
\begin{align}
\langle A  \rangle_\beta  &=\Tr(\rho A),
\nonumber\\
\langle A,B\rangle_{\beta}   
&=\Tr\!\left(\rho^{1/2}A^\dagger\rho^{1/2}  B\right),
\qquad
\norm{A}_{\beta}^{2} =\langle A,A    \rangle_{\beta},
\label{eq:KMS-inner}
\end{align}
with the centring operation $\cent A$ defined in Eq.~\eqref{eq:C-def}. It is convenient to write
\begin{align}
M_{xy}(t)  &\equiv  \cent{[w_x(t),v_y]}, \nonumber\\
C_{xy}(t)  &\equiv \norm{M_{xy}(t)}_\beta^2 .
\label{eq:MC-def}
\end{align}
In a translationally invariant system, $C_r(t) = C_{0r}(t)$.

Next, we consider the symmetrically regularised finite-temperature version of the connected dOTOC introduced in Ref.~\cite{Kukuljan:2017xag},
\begin{align}
D^{(L)}(t) &\equiv \frac{1}{L} \norm{\cent{[W(t),V]}}_{\beta}^{2}
\nonumber\\
&=  \frac{1}{L}\left[
\norm{[W(t),V]}_{\beta}^{2}
-  \abs{\langle[W(t),V]\rangle_\beta}^{2}
\right].
\label{eq:dOTOC-def}
\end{align}
At $\beta=0$, the trace of a commutator vanishes and Eq.~\eqref{eq:dOTOC-def} coincides with the original dOTOC. We keep the finite-volume label $D^{(L)}$ when useful and write $D(t) =\lim_{L\to\infty}D^{(L)}(t)$ for the thermodynamic limit, whenever it exists. The bounds below also control the thermodynamic limsup without assuming convergence.

There is a simple way to expand the dOTOC. Define the spatially integrated local OTOC
\begin{equation}
\bar{C}^{(L)}(t)=\frac1L\sum_{x,y}C_{xy}(t).
\label{eq:Csum-def}
\end{equation}
Translation invariance gives, in the thermodynamic limit,
\begin{equation}
\bar{C}(t) =\sum_r C_r(t).
\label{eq:Csum-translation}
\end{equation}
The full dOTOC is then 
\begin{equation}
D^{(L)}(t) 
= \bar{C}^{(L)} (t) + \frac{1}{L} \!\sum_{(x,y)\neq(x',y')}
\langle M_{xy}(t), M_{x'y'}(t)\rangle_\beta.
\label{eq:D-C-cross}
\end{equation}
Thus, $\bar  C$ sums the local pair contributions, while $D$ also contains their cross terms. If $C_r (t)$ reaches an order-one plateau over a region of radius $R(t)$, and the total exterior contribution is $O(R(t)^d)$, then $\bar C(t)=\Theta(R(t)^d)$.  We will refer to this spatial volume as the scrambling volume.  The dOTOC contains, in addition, the zero-momentum coherence between different local pairs.  Related uses of global OTOCs as probes of operator growth were emphasised in Refs.~\cite{Rakovszky:2017qit,ZhouSwingle2023}.

In Ref.~\cite{Kukuljan:2017xag}, we proved a general bound $D(t)\lesssim t^{3d}$ and used the resulting non-exponential late-time behaviour of the dOTOC to motivate the notion of weak quantum chaos. The results here sharpen this bound and also clarify the origin of the power-law growth. Under the assumptions stated below, we find
\begin{equation}
D(t)  \lesssim t^d.
\label{eq:intro-main-bound}
\end{equation}
Here and throughout the paper, notation $A(t)\lesssim B(t)$ means the following: $\exists K>0$, such that $A(t) < K B(t)$ for sufficiently large $t$.\footnote{We also use standard asymptotic notation. In particular, $f=O(g)$ means $|f|\le K|g|$ for some constant $K$, $f=o(g)$ means $f/g\to0$, and $f=\Theta(g)$ means $K_1|g|\le|f|\le K_2|g|$ for some constants $0<K_1\le K_2$. Thus, $O(1)$ denotes a bounded quantity, $o(t)$ a quantity growing slower than linearly, and $\Theta(t^d)$ a quantity growing exactly as $t^d$ up to constants.} Real-time locality then also gives the same bound for the spatially integrated local commutator-squared OTOC, i.e., $\bar C(t)\lesssim t^d$.  An expanding region of order-one local commutators, with controlled exterior weight, is a sufficient mechanism for saturating the $t^d$ bound on $\bar C$. In the disorder-averaged examples below, this interpretation also applies to $D$ because $\E_J D=\E_J\bar C$. This provides a natural late-time measure of extensive scrambling after the local exponential growth has ended.  In the examples below, the same dOTOC has a Lyapunov regime before the scrambling time, and continues to grow through the spatial proliferation of the saturated commutator afterwards. This gives a concrete interpretation of the weak quantum chaos introduced in Ref.~\cite{Kukuljan:2017xag}. Two explicit many-body models are then studied at the end of this work in which such behaviour can be observed.  The Brownian model realises this crossover explicitly in $d=1$ within the deterministic approximation, while the corresponding late-time result for the SYK chain is conditional on the all-orders eikonal approximation discussed below.

The rest of the paper is organised as follows. In Section~\ref{sec:one-volume}, we prove the one-volume bounds on $\bar C$ and $D$. In Section~\ref{sec:lyapunov-volume}, we discuss the evolution from a pre-scrambling Lyapunov regime to post-scrambling volume growth and derive a spatially integrated version of the MSS bound. Section~\ref{sec:examples} contains the Brownian coupled-cluster and SYK-chain examples. The technical proofs and model details are collected in the appendices.

\section{One-volume bounds from locality and equilibrium clustering}
\label{sec:one-volume}

In order to establish our main bounds on the growth of OTOCs, we work with a finite periodic lattice  $\Lambda=(\mathbb{Z} /\ell \mathbb{Z})^d$, with $L=|\Lambda|=\ell^d$.  We assume that the Hamiltonian $H=\sum_{x\in\Lambda} h_x$ is generated by a uniformly bounded finite-range interaction $h_x$, or more generally, by a short-range interaction for which a standard exponential Lieb--Robinson bound holds~\cite{LiebRobinson1972,NachtergaeleSims2006}.  Furthermore, local densities $h_x$ are assumed to have a fixed-size finite support (containing $x$).  For discussing the full dOTOC, we additionally assume uniform Gibbs clustering which we write in the following form.
There are constants $K_\beta,c_{\rm cl}>0$ and a fixed polynomial $P_{\rm cl}(X,Y)$ in the sizes and diameters of the supports $X,Y\subset \Lambda$ such that, for all finite volumes and local operators $a_X$ and $b_Y$,
\begin{align}
& \abs{\langle a_X b_Y   \rangle_\beta
-\langle a_X \rangle_\beta\langle b_Y  \rangle_\beta}
\nonumber\\
&\le
 K_\beta P_{\rm cl}(X,Y)
\norm{a_X}_\infty\norm{b_Y}_\infty
  e^{-c_{\rm cl}\dist(X,Y)}.
\label{eq:Gibbs-clustering-main}
\end{align}
For the local densities $w_x$ and $v_y$ themselves, whose supports have fixed size, this polynomial factor is a constant. However, we retain it here because later, we apply the clustering bound to Lieb--Robinson truncations whose supports grow with the truncation radius. In $d=1$,  where Eq.~\eqref{eq:Gibbs-clustering-main} follows from Araki's theorem~\cite{Araki1969}, no such polynomial is needed.  In the high-temperature regime of Ref.~\cite{Kliesch2014}, however, the prefactor of the clustering bound does depend on the supports.  This is why we keep $P_{\rm cl}$ explicit.

\begin{theorem}[One causal volume]
\label{thm:one-volume}
Under the assumptions above, there are constants $K,K_\beta,v_{\rm LR},R_0$, independent of $L$ and $t$, such that with
\begin{equation}
R_{\rm LR}(t)=R_0 +  v_{\rm LR}|t|+1 ,
\label{eq:RLR-main}
\end{equation}
one has
\begin{align}
\bar{C}^{(L)} (t)  &\le K\min\{R_{\rm LR}(t)^d ,  L \},
\label{eq:Csum-bound-main}\\
D^{(L)}(t)  &\le K_\beta \min\{R_{\rm LR}(t)^d ,L\}.
\label{eq:dOTOC-bound-main}
\end{align}
The first bound only uses real-time locality. The dOTOC bound additionally uses the clustering condition~\eqref{eq:Gibbs-clustering-main}.
\end{theorem}

The bound on the centred local commutator is immediate. In particular, the Lieb--Robinson bound gives
\begin{align}
\norm{\cent{[w_x(t),v_y]}}_\beta
&\le \norm{[w_x(t),v_y]}_\infty
\nonumber\\
&\le K e^{-\mu_{\rm LR}[\dist(x,y)-v_{\rm LR} |t| - R_0]_+}.
\label{eq:pair-LR-main}
\end{align}
For fixed $x$, summing its square over $y$ therefore gives the volume of one causal ball, which proves Eq.~\eqref{eq:Csum-bound-main}.

The full dOTOC requires a different organisation because it also contains cross terms between distinct local commutators. The first ingredient is the thermal $L_4$ estimate
\begin{equation}
\norm{\cent{[A,B]}}_\beta^2
\le4\sqrt{\langle A^4\rangle_\beta\langle B^4\rangle_\beta},
\label{eq:L4-main}
\end{equation}
for Hermitian $A$  and $B$. The second is a fourth-moment consequence of ordinary clustering. If $a_x$ are uniformly bounded, centred, fixed-support Hermitian observables and $A_S=\sum_{x\in S}a_x$, then
\begin{equation}
\langle A_S^4\rangle_\beta\le K_\beta |S|^2.
\label{eq:fourth-moment-main}
\end{equation}
Both statements are proved in Appendices~\ref{app:thermal} and \ref{app:clustering}.

To apply these estimates, we partition the lattice of source positions $x$ into $d$-dimensional cubes $B$ of side comparable to $R_{\rm LR}(t)$ and write
\begin{equation}
W_B =\sum_{x\in B} (w_x-\langle w_x\rangle_\beta\one),
\qquad
M_B(t)=\cent{[W_B(t),V]}.
\label{eq:block-def-main}
\end{equation}
For each such block, locality allows the target sum to be restricted, up to an exponentially small Lieb--Robinson tail, to the centred operator $V_{S_B}=\sum_{y\in S_B}(v_y-\langle v_y\rangle_\beta\one)$ on an $O(R_{\rm LR}(t))$ enlargement $S_B$ of $B$. For this local part, Eqs.~\eqref{eq:L4-main} and \eqref{eq:fourth-moment-main}, together with thermal stationarity, give
\begin{align}
\norm{M_B(t)}_\beta^2
&\lesssim
\sqrt{\langle W_B(t)^4\rangle_\beta
       \langle V_{S_B}^4\rangle_\beta}
\nonumber\\
&=
\sqrt{\langle W_B^4\rangle_\beta
       \langle V_{S_B}^4\rangle_\beta}
\lesssim |B|R_{\rm LR}(t)^d.
\label{eq:block-bound-main}
\end{align}
Since $|B|=\Theta(R_{\rm LR}(t)^d)$, the squared norm of a causal-size block is $O(R_{\rm LR}(t)^{2d})$. There are $O(L/R_{\rm LR}(t)^d)$ such blocks. Moreover, because both $B$ and $S_B$ have linear size $O(R_{\rm LR}(t))$, the enlarged support associated with any given block overlaps those of only $O(1)$ neighbouring blocks. The contribution of these nearby block pairs is therefore $O(LR_{\rm LR}(t)^d)$. Covariances between more distant blocks are summable and obey the same overall scaling. Dividing by $L$ then leads to the inequality given in Eq.~\eqref{eq:dOTOC-bound-main}.

The symmetric thermal clustering used in this last step follows from Eq.~\eqref{eq:Gibbs-clustering-main} and real-time locality by interpolation in the thermal strip. The proof is given in Appendix~\ref{app:clustering}. The finite-size regime $R_{\rm LR}(t)^d\gtrsim L$ is treated separately in Appendix~\ref{app:one-volume}, where the entire lattice can be used as one block.

The block counting can be summarised as follows. A block contains $O(R_{\rm LR}(t)^d)$ sources and its relevant target region contains $O(R_{\rm LR}(t)^d)$ sites. Equations~\eqref{eq:L4-main} and \eqref{eq:fourth-moment-main} give a squared block commutator of order $R_{\rm LR}(t)^{2d}$. There are $O(L/R_{\rm LR}(t)^d)$ such blocks, and only $O(1)$ neighbouring blocks have overlapping causal regions. The total is therefore $O(LR_{\rm LR}(t)^d)$. A sitewise estimate would instead lose one extra factor of $R_{\rm LR}(t)^d$. The difference comes from the fourth-moment bound for the full causal block. Appendix~\ref{app:one-volume} also gives a simple example showing that positivity and pairwise locality by themselves do not supply this missing constraint.

Uniform exponential clustering \eqref{eq:Gibbs-clustering-main} holds for one-dimensional finite-range translation-invariant spin systems at every finite temperature~\cite{Araki1969,Bluhm2022}, and at sufficiently high temperature   $\beta < \beta_{\rm c}$ in arbitrary dimension~\cite{Kliesch2014}. Theorem~\ref{thm:one-volume} therefore applies in those regimes and, more generally, wherever Eq.~\eqref{eq:Gibbs-clustering-main} holds.  Some restriction excluding non-clustering equilibrium states is necessary for   a volume-uniform statement.   For example, consider the ferromagnetic Ising model in $d\ge2$ below $T_c$, on the periodic lattice.  Its finite-volume Gibbs state is unique and symmetric, but its clustering constant is not uniform in $L$.  Taking $w_x=\sigma_x^1$, $v_x= \sigma_x^2$, $W =\sum_xw_x$ and $V =\sum_xv_x$, we find $[W,V]=2i\sum_x \sigma_x^3$ at $t=0$. Hence, $D^{(L)}(0)=4\operatorname{Var}_\beta(\sum_x\sigma_x^3)/L$, which is extensive when $\operatorname{Var}_\beta(\sum_x\sigma_x^3)=\Theta(L^2)$.  Here, the variance is defined as $\operatorname{Var}
_\beta(A) \equiv
\langle A^2\rangle_\beta - \langle A\rangle_\beta^2$.

We note that the power $d$ in the bound on $\bar C$ cannot in general be improved. If the local squared commutator approaches an order-one value behind a ballistic front, the number of sites  contributing to $\bar C$ grows as $t^d$, and hence $\bar C(t)\sim t^d $. The same behaviour can also occur for the full dOTOC. In the Brownian model below we find, after disorder averaging, $D(t)= 4 v_B(N) t + o(t)$ in one spatial dimension. We also find a corresponding late-time linear behaviour in the SYK  chain, subject to the all-orders eikonal assumption discussed in Section~\ref{sec:SYK-main}.

Taking the thermodynamic limit at fixed time whenever it exists (or taking the limsup otherwise)  in Eqs.~\eqref{eq:Csum-bound-main} and \eqref{eq:dOTOC-bound-main}, the rigorous temporal growth bounds are
\begin{equation}
\bar{C}(t)  \lesssim t^d,
\qquad
D (t) \lesssim t^d.
\label{eq:hierarchy}
\end{equation}
Importantly, real-time locality is sufficient for the first bound, while equilibrium clustering is additionally required for the second. All constants are uniform in spatial volume for a fixed microscopic model, and at fixed $N$.

\section{From Lyapunov growth to scrambling volume}\label{sec:lyapunov-volume}

We now ask how the one-volume regime of growth with $t^d$ is reached in a chaotic system.  We begin with $\bar C(t)=\sum_r C_r(t)$, for which the spatial interpretation is simpler and more direct. The dOTOC instead  contains the cross terms in Eq.~\eqref{eq:D-C-cross}, so its crossover depends on their dynamics as well. In the two examples of Section~\ref{sec:examples}, exact ensemble symmetries remove these terms after disorder averaging and give $\E_JD=\E_J\bar C$.

For the analysis, we note that the order in which the late-time and the large-$N$ limits are taken is important. We keep $N$ large and finite while following the evolution through $t_*(N)$. The regime $t\ll t_*$ is referred to as pre-scrambling, $t  \sim t_*$ is the onset of local saturation, and $t> t_*$ probes the subsequent growth of the saturated region. Taking $N= \infty$ first at fixed time sends $t_*$ to infinity.

A convenient way to describe the pre-scrambling spatial profile is in terms of the velocity-dependent Lyapunov exponent~\cite{Khemani2018,MezeiSarosi2020}. Here,  $\lambda(v)$ is the exponential growth rate measured along the ray $r=vt$. Writing $v=r/t$, suppose that before local saturation,
\begin{equation}
C_r (t)= \eps_N b(v,t)e^{t\lambda(v)},
\qquad
\lambda(v) \le\lambda(0) =\lambda_L,
\label{eq:vdle-main}
\end{equation}
where $b(v, t)$ is a positive subexponential prefactor, $\ln b(v,t)  = o(t)$  uniformly on the bounded velocity range considered, which remains nonzero near $v=0$, and the weight outside a bounded velocity range is exponentially subleading. Along any sequence with $t\to\infty$ and $N\to\infty$ simultaneously, such that $\eps_{N} e^{\lambda_L t} \to 0$  and the full local profile, including $b(v,t)$, tends uniformly to zero,  we have
\begin{equation}
\frac{1}{t} \ln \frac{\bar{C}(t)}{\eps_{N}}
\to \lambda_L.
\label{eq:rate-preserve}
\end{equation}
Thus, spatial integration preserves the Lyapunov exponent. It can, however, generate a polynomial prefactor. 

If the maximum of $\lambda(v)$ at $v=0$ is smooth and non-degenerate, we may write
\begin{equation}
\lambda(v)= \lambda_L-\frac{1}{2} v^{\mathsf T} \mathsf H v
+O(|v|^3),
\qquad
\mathsf H > 0,
\end{equation}
where $\mathsf H$ is the Hessian of $-\lambda$ at $v=0$. Here, we assume we are working in $d$ spatial dimensions so that $r$ and $v = r/t$ are vectors. If $b(v,t)$ approaches a nonzero constant $b_0$ in the region $|r|=O(\sqrt{t})$ that dominates the spatial sum, expanding $\lambda(v)$ around its maximum gives
\begin{equation}
\bar{C}(t)\sim
\eps_N e^{\lambda_Lt}
\, b_0\frac{(2\pi t)^{d/2}}{\sqrt{\det\mathsf H}}.
\label{eq:laplace-main}
\end{equation}
The spatial integration does not change the exponential growth rate. The additional factor $t^{d/2}$ arises because the dominant contribution comes from a region of width $|r|=O(\sqrt{t})$ in each spatial direction. A time-dependent prefactor in the local OTOC can modify this power. For example, a heat-kernel prefactor $t^{-d/2}$ cancels it exactly, as happens in both many-body examples studied below.

The onset of saturation is a separate regime and need not immediately resemble a rigid ballistic front. To see this, consider the following simple profile:
\begin{equation}
C_r(t)= C_*\min\!\left\{1,\eps_Ne^{t\lambda(r/t)}\right\},
\qquad
\eps_N=e^{-\lambda_Lt_*},
\label{eq:clipped}
\end{equation}
where $C_*>0$ is an order-one value setting the `height' of the plateau. Its saturation contour is determined by
\begin{equation}
t\,[\lambda_L-\lambda(r/t)]=\lambda_L(t-t_*).
\label{eq:sat-contour}
\end{equation}
For the quadratic behaviour of $\lambda(v)$ near $v=0$, consider times just after the onset of saturation, $0<t-t_*\ll t_*$. Expanding the left-hand side of Eq.~\eqref{eq:sat-contour} gives
\begin{equation}
\frac{1}{2t} r^{\mathsf T}\mathsf H r
+ O\left(\frac{|r|^3}{t^2}\right)
= \lambda_L (t-t_*).
\label{eq:quadratic-crossover}
\end{equation}
At $t=t_*$, the local OTOC at the origin first becomes order one, although the pre-scrambling profile is already spatially extended around it.  At slightly later times, increasingly distant points cross the same saturation threshold.  For the quadratic form of $\lambda(v)$ near $v= 0$, the saturation condition gives $|r|^2\sim t_*(t-t_*)$, up to direction-dependent constants set by $\mathsf H$. Thus, the characteristic size of the newly saturated region grows as $\sqrt{t_*(t-t_*)}$, and its volume as $[t_*(t-t_*)]^{d/2}$.  This regime describes the progressive saturation of a pre-scrambling profile that is already spatially extended.  At later times, the boundary of the order-one region may instead approach a ballistic front with a finite velocity.

Deep after scrambling, considerably less information is needed. Suppose there are fixed velocities $0 < v_- <v_+$ such that $C_r(t)$ is bounded from below by a positive constant for $|r|\leq v_-(t-t_{*})$, remains uniformly bounded for $|r|\leq v_+(t-t_{*})$, and the total weight outside the outer ball is $o((t-t_{*})^d)$. Then,
\begin{equation}
\bar{C}(t)=\Theta((\Delta t)^d),
\qquad \Delta t=t-t_*\to\infty.
\label{eq:late-volume}
\end{equation}
This is precisely the growth of one ballistic volume. It also shows how the geometric upper bound $\bar{C}(t)\lesssim t^d$ proved in Section~\ref{sec:one-volume} can be saturated. See Appendix~\ref{app:largeN} for elementary large-deviation and volume estimates underlying these statements.

\subsection{Spatial integration and the MSS-type bound}

There also exists a natural MSS-type  constraint on the pre-scrambling spatial profile. The key observation is that a positive normalised spatial average of functions mapping the MSS strip into the unit disk has the same property. Let $F_r(z)$ be the normalised correlator of Eq.~\eqref{eq:F-def}, analytically continued to the MSS strip
\begin{equation}
|\Im z| <\frac{\beta}{4},
\qquad |F_r(z)| \le  1.
\label{eq:MSS-strip-local}
\end{equation}
Here, we impose the exact full-strip bound in Eq.~\eqref{eq:MSS-strip-local} as an additional assumption, and it  does not in general follow from the normalisation in Eq.~\eqref{eq:F-def}. The usual MSS argument uses a half-strip and includes normalisation errors. Their spatially weighted contribution must be controlled before taking the limits below (see Ref.~\cite{Maldacena:2015waa}).  For positive weights $a_r$, define
\begin{equation}
F_a(z)=
\frac{\sum_r a_rF_r(z)}{\sum_r a_r}.
\label{eq:positive-smearing}
\end{equation}
The unit disk is convex, so the positive normalised average $F_a(z)$ is analytic in the same strip and still satisfies $|F_a(z)|\leq1$. Mapping the strip to the unit disk and applying the Schwarz--Pick theorem therefore gives, on the real axis,
\begin{equation}
\frac{|\partial_tF_a(t)|}{1-|F_a(t)|^2}
\le\frac{\pi}{\beta}.
\label{eq:spatial-SP-main}
\end{equation}
Consequently,
\begin{equation}
\left|  \partial_t \sum_r a_r[1-F_r(t)]  \right|
\le\frac{2\pi}{\beta}
\left|  \sum_r a_r[1-F_r(t)] \right|.
\label{eq:spatial-MSS-main}
\end{equation}
The result is independent of the number of spatial points $r$ that are included. Using the integrated form of this inequality in finite volume and then taking the thermodynamic limit avoids any assumption about convergence of time derivatives.  For Hermitian probes, $F_r(t)$ is real on the real axis and the assumed bound (normalisation) $|F_r(t)|\le 1$ makes $1-F_r(t)$ non-negative.

For exponential weights $a_r = e^{\mu \cdot r}$, with real $\mu$ in the domain where the spatial transform converges, the same argument constrains spatial propagation as well as local growth. The reason is that the exponential weight tilts the spatial profile.  Let $\lambda_F(v)$ denote the velocity-dependent exponent of the MSS four-point deficit $1-F_r(t)$.  Note that this is not the same as the exponent $\lambda(v)$ of the squared commutator in Eq.~\eqref{eq:vdle-main}.  If the normalised large-$N$ four-point deficit has the  form
\begin{equation}
1 - F_r(t)
=
\eps_N \exp\!\left[ t \lambda_F(r/t) + o(t) \right],
\label{eq:linearised-LD-main}
\end{equation}
and the tails remain exponentially negligible after each admissible tilt,  then along a ray $r = v  t$ the weight contributes an additional factor $e^{t \mu \cdot v}$. The weighted spatial sum is therefore governed at late times by the largest tilted exponent,
\begin{equation}
\sup_v \left[ \lambda_F(v) + \mu \cdot v \right].
\end{equation}
This is the Legendre  transform of $-\lambda_F(v)$, evaluated at $\mu$. Equivalently, when $\lambda_F(v)$ is smooth and concave, the weight $\mu$ selects the ray velocity $v_\mu$ satisfying
\begin{equation}
\nabla_v \lambda_F(v_\mu) = -\mu.
\end{equation}
Applying the MSS bound to the exponentially weighted sum then gives
\begin{equation}
\sup_v \left[ \lambda_F(v) + \mu \cdot v \right]
\le \frac{2\pi}{\beta},
\label{eq:Legendre-MSS-main}
\end{equation}
for every admissible $\mu$. Thus, each exponential weight gives a linear bound of the following form:
\begin{equation}
\lambda_F(v)
\le \frac{2\pi}{\beta} - \mu \cdot v.
\end{equation}
Optimising over the allowed values of $\mu$ gives the strongest bound at each velocity. Geometrically, these inequalities form a family of supporting planes whose envelope bounds the full velocity-dependent growth exponent $\lambda_F(v)$.  We give a more detailed derivation in Appendix~\ref{app:spatial-MSS}.

For each admissible tilt $\mu$, the bound
$\lambda_F(v) \le 2\pi /\beta-\mu \cdot v$ defines a supporting `hyperplane'
in $(v,\lambda)$ space, passing through
$(v,\lambda)=(0, 2\pi/\beta)$.  If the admissible tilts contain a ball $|\mu|<\mu_*$, their envelope implies $\lambda_F(v)\le 2 \pi/\beta-\mu_*|v|$. Here, the normalised large-$N$ deficit must have a nonzero domain of convergence. We note that this does not follow from finite-$N$ locality and clustering alone. In general, the strongest envelope may be anisotropic. This, however, has an interesting consequence for four-point functions with maximal growth, i.e., when $\lambda_F(0) = 2\pi / \beta$.  In such systems, $\lambda_F(v)$ cannot have a smooth, non-degenerate maximum at $v =0$.  If it did, then the same expansion that leads to Eq.~\eqref{eq:laplace-main} would give the tilted growth rate $\lambda_F(0)  + \frac{1}{2} \mu^{T} \mathsf H_F^{-1} \mu+O(|\mu|^3)$, where $\mathsf H_F$ is the Hessian of $-\lambda_F$ at $v=0$.  For small admissible $\mu\neq 0$, this exceeds $2\pi/\beta$ and contradicts Eq.~\eqref{eq:Legendre-MSS-main}.   Hence, under these assumptions and conditions, the exponent must fall at least linearly from its maximum, $\lambda_F (v) \le \lambda_F(0) - \mu_*|v|$, in agreement with the velocity-dependent form of Ref.~\cite{MezeiSarosi2020}.  Conversely, a smooth maximum of $\lambda_F(v)$ at $v=0$ is only possible when $\lambda_F(0) < 2\pi/\beta$. The pole-free mode of the SYK chain in Section~\ref{sec:examples} has exponent $2\pi\nu/\beta$ with $\nu < 1$.

\section{Solvable examples and a numerical crossover}\label{sec:examples}

We now study two models in which the cross terms in $D(t)$ do not have to be bounded. In both, an exact symmetry of the disorder ensemble removes every off-diagonal pair term after averaging, so that the relation between $D$ and $\bar C$ is known exactly. The Brownian model then gives the full crossover in the deterministic approximation. The SYK chain gives a finite-temperature static example with a controlled early Lyapunov regime and, assuming the standard late-time eikonal form of the OTOC, an explicit post-scrambling result.

\subsection{Brownian model: the `deterministic' crossover}

The Brownian coupled-cluster model introduced by Xu and Swingle~\cite{XuSwingle2019} consists of a one-dimensional chain of clusters, each containing $N$ spin-$1/2$ degrees of freedom. The spins interact through random two-spin couplings, both within each cluster and between neighbouring clusters. More explicitly, the time-dependent Hamiltonian has the form
\begin{equation}
H(t) = \sum_x H_x(t)  + \sum_x H_{x,x+1}(t),
\end{equation}
with
\begin{align}
H_x(t) &=
\sum_{a<b} \sum_{\alpha,\beta=x,y,z} J^{\alpha\beta}_{x;ab}(t) \,  \sigma^\alpha_{x,a}\sigma^\beta_{x,b},
\nonumber\\
H_{x,x+1}(t)
&= g \sum_{a,b}
\sum_{\alpha,\beta=x,y,z}
\widetilde J^{\alpha\beta}_{x;ab}(t)\,
\sigma^\alpha_{x,a}\sigma^\beta_{x+1,b}.
\end{align}
Here, $a,b=1,\ldots,N$ label spins within a cluster. We take $N\ge 2$ and $g>0$. For the unordered on-site pairs $a<b$ used above, the independent centred Gaussian white noises are normalised by
\begin{align}
\E_J[J_I(t)J_{I'}(t')] &=\frac{\delta_{II'}\delta(t-t')}{8(N-1)},\nonumber\\
\E_J[\widetilde J_I(t)\widetilde J_{I'}(t')]&=\frac{\delta_{II'}\delta(t-t')}{16N},
\end{align}
with zero mixed covariance. Here, $I$ denotes the full set of site, spin and Pauli indices. The factor $g$ remains explicit in $H_{x,x+1}$. The stochastic evolution is defined as the limit of independent unitary increments $U(t+\mathrm dt) = e^{-i\,\mathrm d K_t}U(t)$, where $\mathrm d K_t$ is obtained from the displayed Hamiltonian by replacing $J_I(t)\mathrm dt$ and $\widetilde J_I(t)\mathrm dt$ with Brownian increments of the corresponding variances. This convention fixes the time units used below~\cite{XuSwingle2019}.  Thus, unlike an ordinary disordered Hamiltonian, the couplings themselves are independently renewed in time.

Averaging over this Brownian dynamics gives a closed evolution equation for the operator-weight distribution. At large $N$, neglecting operator-weight fluctuations gives a deterministic nonlinear equation describing exponential growth, propagation of a ballistic front and eventual local saturation. Because the couplings are white noise in time, this model lies outside the static  interaction assumption of Theorem~\ref{thm:one-volume}. In this case, we use it to provide a solvable illustration of the one-volume growth.

At infinite temperature ($\beta = 0$), we take $w_x= \sigma^z_{x,1}$ and $v_x= \sigma^z_{x,2}$. A local Pauli conjugation flips any chosen source or target pair commutator while leaving the Brownian measure invariant. Every cross term between distinct pairs is therefore odd under an exact ensemble symmetry, so
\begin{equation}
\E_J D^{(L)}(t)  = \E_J \bar{C}^{(L)} (t)
\label{eq:BCC-equality}
\end{equation}
for every $L$, $N$,  and $t$. This is an equality of disorder averages (a short proof is presented in Appendix~\ref{app:BCC}).

We now use the operator-growth results of Ref.~\cite{XuSwingle2019}. Their closed cluster-weight description applies after a short $O(1)$ microscopic transient, well before the logarithmic scrambling time. In this regime, the linearised dynamics determines the spatial sum before saturation, while at late times the saturated travelling front determines its growth. Since Eq.~\eqref{eq:BCC-equality} holds exactly after Brownian averaging, the same asymptotics apply to the full dOTOC:
\begin{align}
\E_J \bar{C}^{(L)}(t) = \E_J D^{(L)}(t) &=   \frac{8}{3N} e^{\gamma t} [1 + o(1)], \nonumber\\
&\qquad\quad 1 \ll t, \quad N^{-1} e^{\gamma t} \to 0, \label{eq:BCC-early-main}\\
\E_J \bar{C}^{(L)}(t) = \E_J D^{(L)}(t) &= 4 v_B(N) t + o(t), \nonumber\\
&\qquad\quad t/t_* \to \infty, \label{eq:BCC-crossover}
\end{align}
where
\begin{align}
\gamma &= 3(1 + g^2), \qquad
t_* \sim
\frac{\ln N}{\gamma}, \nonumber\\
v_B &= \min_{\kappa > 0}
\frac{3[1 + g^2 \cosh \kappa]}{\kappa}.
\label{eq:BCC-parameters}
\end{align}
where $\kappa>0$ is the spatial decay exponent of the
linearised leading edge. Moreover, $v_B$ is the deterministic velocity. At large finite $N$, $v_B(N)=v_B+O(1/\ln^2N)$~\cite{XuSwingle2019}. The late-time limit is taken with $L\to\infty$ and $v_B(N)t\ll L$, before wraparound.  The early-time expression in Eq.~\eqref{eq:BCC-early-main}  comes from the dilute pre-scrambling operator-weight regime.  The late-time expression in Eq.~\eqref{eq:BCC-crossover}  counts the expanding saturated interval. To reach the late-time regime, one needs to keep $N$ large and finite, follow the dynamics past $t_*\sim(\ln N)/\gamma$, and then take the corresponding large-$N$ scaling limit.  Thus, both $\bar C$ and $D$ can be found explicitly in the two asymptotic regimes, with equal disorder averages.

The crossover can be visualised without introducing an ad hoc interpolation. The deterministic large-$N$ lattice equation obtained from the Brownian operator-weight master equation is~\cite{XuSwingle2019}
\begin{equation}
\dot\phi_r(t)
=\bigl[3-4\phi_r(t)\bigr]
\left[\phi_r(t)+\frac{g^2}{2}
\bigl(\phi_{r-1}(t)+\phi_{r+1}(t)\bigr)\right],
\label{eq:BCC-mean-field}
\end{equation}
with the localised initial condition $\phi_r (0) = \delta_{r0}/N$.  Here, $\phi_r$ is the deterministic approximation to the mean operator weight per spin in cluster $r$. After the microscopic internal-relaxation transient, the corresponding approximation to the disorder-averaged spatially integrated OTOC is
\begin{equation}
\E_J \bar{C}(t) \simeq
\frac{8}{3} \sum_r \phi_r(t).
\label{eq:BCC-Cmf}
\end{equation}
Together with Eq.~\eqref{eq:BCC-equality}, this also approximates $\E_J D(t)$. For the $1/N$-dependent initial condition (the `seed'), we have not established independently that this approximation remains valid throughout the crossover.  Equation~\eqref{eq:BCC-mean-field} is the noiseless $N=\infty$ drift equation, and, note that, finite-$N$ front fluctuations are not represented by it and shift the selected front velocity by the Brunet--Derrida-type $O(1/\ln^2N)$ correction discussed in Ref.~\cite{XuSwingle2019}. At earlier times, the fixed-spin commutator contains additional short-time structure, so we do not use the large-$N$ operator-weight equation in that regime. Linearising Eq.~\eqref{eq:BCC-mean-field} around $\phi_r = 0$ gives the growth rate $\gamma$ and selects the lattice front velocity $v_B$ in Eq.~\eqref{eq:BCC-parameters}, while the nonlinear term drives the local profile to the saturated value $\phi_r = 3/4$ behind the front. Figure~\ref{fig:BCC-crossover} shows the resulting deterministic approximation to the behaviour of $\mathbb E_J D(t)$ and $\mathbb E_J \bar{C} (t)$ in the pre- and post-scrambling regimes.

\begin{figure*}[t]
\centering
\begin{minipage}[t]{0.485\textwidth}
\centering
\includegraphics[width=\linewidth]{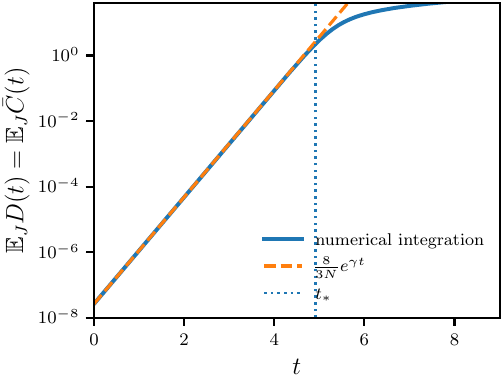}
\end{minipage}\hfill
\begin{minipage}[t]{0.485\textwidth}
\centering
\includegraphics[width=\linewidth]{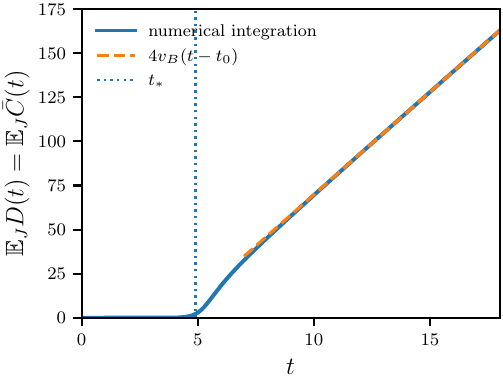}
\end{minipage}
\caption{Numerical Brownian crossover from Lyapunov growth to scrambling-volume growth. We integrate Eq.~\eqref{eq:BCC-mean-field} on $L=801$ sites with $g=1/2$, $N=10^8$, and $\phi_r(0) = \delta_{r0}/N$.  The plotted quantity is $(8/3)\sum_r\phi_r(t)$.  The axis labels denote the deterministic approximations to $\mathbb E_JD$ and $\mathbb E_J\bar C$ in Eqs.~\eqref{eq:BCC-Cmf} and \eqref{eq:BCC-equality}. Finite-$N$ front fluctuations are not included in Eq.~\eqref{eq:BCC-mean-field}. {\bf Left:} On a logarithmic scale, the numerical result follows $8e^{\gamma t} / (3N)$ through the pre-saturation regime; the leading logarithmic scrambling estimate is $t_* = (\ln N) / \gamma=4.912$. {\bf Right:} After saturation, the same curve approaches one-volume growth.  The dashed line has the analytic large-$N$ slope $4v_B =11.631$ and an effective front delay $t_0 =4.000$ obtained from a one-parameter fit with the slope fixed. An unconstrained fit over $12\le t\le 18$ gives slope $11.613$, which differs from the analytic large-$N$ slope $4v_B$ by a relative amount of $1.5 \times 10^{-3}$.}
\label{fig:BCC-crossover}
\end{figure*}

\subsection{SYK chain: a thermal Lyapunov regime and a strong-coupling late-time prediction}
\label{sec:SYK-main}

Next, we consider the finite-temperature SYK chain introduced by Gu, Qi and Stanford~\cite{Gu:2016oyy}, using the large-$q$ finite-coupling analysis of Choi, Haehl, Mezei and S\'{a}rosi~\cite{CHMS2023}. Unlike in the Brownian model considered above, the Hamiltonian is now static and the thermal state is taken at inverse temperature $\beta$. We consider a chain of $L$ sites, each containing $N$ Majorana fermions $\chi_{x, a}$, where $x$ labels the site and $a=1,\ldots,N$ the flavour, with
\begin{equation}
\left\{ \chi_{x,a}, \chi_{y,b} \right\}
=
\delta_{xy} \delta_{ab}.
\end{equation}
For even $q$, the Hamiltonian contains random $q$-body interactions within each site and random nearest-neighbour interactions involving $q/2$ Majoranas from each of the two sites,
\begin{equation}
H
=
\sum_x
\left(
H_x^{(0)} + H_{x,x+1}^{(1)}
\right),
\end{equation}
with
\begin{align}
H_x^{(0)} &= i^{q/2} \sum_{a_1<\cdots<a_q}
J^{(0)}_{x;a_1\ldots a_q}\,
\chi_{x,a_1}\cdots\chi_{x,a_q},
\nonumber\\
H_{x,x+1}^{(1)} &= i^{q/2} \sum_{\substack{
a_1<\cdots<a_{q/2}\\
b_1<\cdots<b_{q/2}
}}
J^{(1)}_{x;a_1\ldots a_{q/2};b_1\ldots b_{q/2}}\,  \nonumber \\
&\times \chi_{x,a_1}  \cdots\chi_{x,a_{q/2}}
\chi_{x+1,b_1}  \cdots\chi_{x+1,b_{q/2}} .
\end{align}
The couplings $J^{(0)}$ and $J^{(1)}$ are independent centred Gaussian random variables with second moments
\begin{align}
\E_J\!\left[
\left(J^{(0)}_{x;a_1\ldots a_q}\right)^2
\right]
&=
\frac{(q-1)!}{N^{q-1}} J_0^2,
\nonumber\\
\E_J\! \left[
\left(J^{(1)}_{x; a_1\ldots a_{q/2};b_1 \ldots b_{q/2}}\right)^2
\right] &=
\frac{[(q/2)!]^2}{qN^{q-1}} J_1^2 .
\end{align}
Here, $J_0$ and $J_1$ set the strengths of the on-site and nearest-neighbour disorder, respectively. For each finite chain and fixed disorder realisation, the Hamiltonian is  finite range and the local Hilbert space is finite. We do not, however,
invoke Theorem~\ref{thm:one-volume} directly for the thermodynamic
Gaussian-disorder ensemble, since its assumptions include volume-uniform
bounds on the local interaction strengths. The asymptotic statements below concern, instead, the disorder-averaged observable in the regime $N\gg q^2\gg1$.

In the large-$q$ limit, one keeps fixed the rescaled couplings
\begin{equation}
\mathcal J_0^2
= \frac{qJ_0^2}{2^{q-1}},
\qquad \mathcal J_1^2  = \frac{qJ_1^2}{2^{q-1}} .
\end{equation}
Following Ref.~\cite{CHMS2023}, we then define
\begin{align}
\mathcal J^2
&= \mathcal J_0^2 + \mathcal J_1^2, \nonumber\\
\eta &= \frac{\mathcal J_1^2}{\mathcal J^2},
\qquad
0 < \eta \le 1,
\nonumber\\
\frac{\pi \nu}{\cos(\pi \nu/2)}
&= \beta \mathcal J,
\qquad   0 < \nu < 1.
\label{eq:SYK-parameters-main}
\end{align}
In terms of these parameters, $\eta$ measures the relative strength of the intersite coupling, while $\nu$ parametrises the temperature and coupling. Strong coupling, $\beta \mathcal J \gg 1$, corresponds to $\nu \to 1$. It will also be convenient to measure real time in units of temperature by defining
\begin{equation}
\tau
=
\frac{2\pi t}{\beta}.
\label{eq:SYK-tau-main}
\end{equation}

For the local probes, let $A$ and $B$ be two disjoint sets of Majorana flavours with
\begin{equation}
|A|
=
|B|
=
s \in 2\mathbb Z,
\qquad
N > 2s,
\qquad
\frac{s}{q}
\to
\Delta > 0,
\end{equation}
in the large-$q$ limit. For any flavour set $S$, define the normalised Hermitian Majorana string
\begin{equation}
o_{S,x}
=
2^{|S|/2}
i^{|S|(|S|-1)/2}
\prod_{a \in S} \chi_{x,a},
\qquad
o_{S,x}^2
=
1.
\end{equation}
For the operators in the OTOCs, we choose
\begin{equation}
w_x  = o_{A,x}, \qquad v_x =  o_{B,x}.
\end{equation}
Since $s$ is even, both probes have even fermion parity. The parameter $\Delta=s/q$ is kept fixed as $q\to\infty$.

For a chain of length $L$, let $D^{(L)}(t)$ and $\bar C^{(L)}(t)$ denote the finite-volume dOTOC and spatially integrated squared commutator defined above. We write $\mathbb E_J$ for the Gaussian disorder average over the static coupling tensors $J$. The thermal expectation value is always evaluated first for the fixed Hamiltonian $H_J$, and only afterwards is the disorder average taken.

The isotropic Gaussian disorder ensemble is invariant under local $SO(N)$ rotations of the Majorana flavours. Since $A$ and $B$ are disjoint and $N>2s$, such a rotation can flip the sign of a single source or target string, while leaving all other strings unchanged.  Every off-diagonal pair term is therefore odd under an exact symmetry of the ensemble, and
\begin{equation}
\mathbb E_J D^{(L)}(t)
=
\mathbb E_J \bar C^{(L)} (t)
\label{eq:SYK-exact-equality}
\end{equation}
exactly at finite $N$, $q$, $L$ and $\beta$, provided $N>2s$. The identity in Eq.~\eqref{eq:SYK-exact-equality},  analogous to that of the Brownian model above, does not rely on the large-$N$ or large-$q$ limit.  Its proof is given in Appendix~\ref{app:SYK}.

We first consider the pre-scrambling regime. The asymptotic formulae of Ref.~\cite{CHMS2023} require
\begin{equation}
N \gg q^2 \gg 1.
\end{equation}
The dynamical asymptotic formulas and logarithmic time scales below are understood within the leading-large-$q$ effective theory, at fixed $\Delta,\nu,\eta$. This hierarchy alone does not make them uniform finite-$q$ asymptotics at growing times. If the exact dimensionless exponent is $\kappa_q(0)=\nu+\delta\kappa_q$, relative $o(1)$ accuracy additionally requires $\tau\delta\kappa_q\to0$, together with control of the other omitted terms. For example, an $O(q^{-1})$ exponent correction requires $\log(N/q^2)/q\to0$ near the logarithmic crossover scales. A useful feature of the present observable is that the maximal-chaos contribution to the ordinary one-string four-point function cancels at linear order. To see why, consider a momentum mode whose real-time dependence is $e^{\kappa(p)\tau}$, where $\kappa(p)$ is the dimensionless growth exponent given explicitly in Eq.~\eqref{eq:app-SYK-kappa}. The two OTOCs entering the symmetrically regularised commutator lie on the two edges of the MSS strip, at $t \pm i\beta/4$, and therefore contribute the factor of $2 \cos( \pi \kappa / 2 )$.  The one-string four-point function considered in Ref.~\cite{CHMS2023} contains precisely the corresponding factor $1/\cos(\pi \kappa/2)$.  Thus, the pole at $\kappa=1$, associated with maximal growth $e^{2\pi t/\beta}$, is cancelled in the  symmetric commutator~\cite{BlakeLiu2021}.

The leading regular mode that remains has dimensionless exponent $\kappa=\nu$ at zero spatial momentum. Its physical Lyapunov exponent is therefore
\begin{equation}
\lambda_L
=
\frac{2\pi \nu}{\beta}.
\label{eq:SYK-lambda-main}
\end{equation}
While higher  eikonal orders are still subleading, this pole-free mode gives
\begin{align}
\mathbb E_J \bar C^{(L)}(t)
= \mathbb E_J D^{(L)}(t)
&=
\mathcal A_{\Delta,\nu,\eta} \frac{q^2}{N} e^{\lambda_L t} [1 + o(1)].
\label{eq:SYK-early}
\end{align}
Here, $\mathcal A_{\Delta,\nu,\eta} \ge 0$ is a finite external-vertex factor determined by the choice of probes and by $\Delta$, $\nu$ and $\eta$. For probes with nonzero overlap with the pole-free mode, $\mathcal A_{\Delta,\nu,\eta} >0$. Its detailed form is, however, not needed to determine either $\lambda_L$ or the prefactor $q^2/N$. The corresponding local profile is Gaussian in space and is given explicitly in Appendix~\ref{app:SYK}.

The symmetric commutator receives a second-order eikonal contribution.  At fixed $1/2<\nu<1$ and $\eta$, its spatial sum grows as $(q^2/N)^2e^{2\nu\tau}/\sqrt{\tau}$, as derived in Appendix~\ref{app:SYK}. Equating it to Eq.~\eqref{eq:SYK-early} gives the formal crossover scale
\begin{equation}
\nu\tau_\times=\ln\frac{N}{q^2}+\frac12\ln\tau_\times+O(1).
\label{eq:SYK-tcross-main}
\end{equation}
For a parametrically separated perturbative second-order regime, one also needs strong coupling. Large $N$ at fixed coupling is, in itself, not sufficient. This does not, however, conflict with the MSS bound, which constrains $F_r$ rather than $C_r$.

It is useful to distinguish the following logarithmic scales. The first is $\tau_\times$ in Eq.~\eqref{eq:SYK-tcross-main}. A second scale is obtained by formally extrapolating the local pole-free Gaussian profile until its value at the origin becomes of order one. We denote this time by $t_{\rm pf}$ and define the corresponding dimensionless time
\begin{equation}
\tau_{\rm pf}\equiv \frac{2\pi t_{\rm pf}}{\beta}.
\end{equation}
The Gaussian prefactor gives
\begin{equation}
\nu \tau_{\rm pf}
= \ln\frac{N}{q^2}
+ \frac{1}{2}\ln \tau_{\rm pf} + O(1).
\label{eq:SYK-tpf-main}
\end{equation}
The term $\frac{1}{2}\ln \tau_{\rm pf}$ comes from the $1/\sqrt{\tau}$ width factor of the local Gaussian saddle.

The far-distance nonlinear energy-momentum tensor front has an extrapolated delay $t_B$, defined by
\begin{equation}
\tau_B
\equiv
\frac{2\pi t_B}{\beta}
=
\ln \frac{N}{q^2}
+
O(1).
\label{eq:SYK-tB-main}
\end{equation}
At fixed $0<\nu<1$ and large $N$,
\begin{equation}
\tau_{\rm pf}
=
\frac{1}{\nu}
\ln \frac{N}{q^2}
+
O(\ln\ln N),
\qquad
\tau_B
=
\ln \frac{N}{q^2}
+
O(1).
\end{equation}
At late times, the front moves away from the origin at an approximately constant speed. If we extend this late-time motion backwards, the front would reach the origin at $\tau=\tau_B$. It is in this sense that $\tau_B$ describes a delay in the propagation of the front.

The large-$q$ thermal two-point function fixes the symmetric thermal norms of the parity-even strings. To infer a late-time plateau, we additionally assume local factorisation of the time-ordered terms, decay of the strip-edge OTOCs, and a subleading connected one-point subtraction. Under these pointwise dynamical assumptions, the two strings chosen above have the conditional plateau
\begin{equation}
C_\infty
=
2
\left[
\cos\!\left( \frac{\pi \nu}{2} \right)
\right]^{4\Delta}
[1 + o(1)].
\label{eq:SYK-plateau-main}
\end{equation}
$C_\infty$ denotes the leading large-$N$, large-$q$ value approached by the local squared commutator deep inside a scrambled region.  The connected one-point subtraction is assumed subleading in this limit for the two disjoint strings. Equation~\eqref{eq:SYK-plateau-main} is therefore a large-$N$, large-$q$ statement and should be distinguished from the exact finite-$N$ identity in Eq.~\eqref{eq:SYK-exact-equality}.  At strong coupling, the plateau is parametrically small, $C_\infty\simeq 2( \pi/\beta\mathcal J)^{4\Delta}$.  For $\Delta >  0$, the norm $\|o_{S}\|^2_\beta$ of a Majorana string vanishes as $\beta \mathcal J\to\infty$.  Hence, the natural unit for the saturated commutator is $\|w\|_\beta^2 \|v\|_\beta^2$, rather than one. In these units, the conditional leading plateau is $2$, as in the Brownian model.

The cancellation of the maximal-chaos pole described above is only a statement about the term linear in that mode. In the coherent energy-momentum tensor sector of the leading eikonal approximation,  the two edges of the MSS contour acquire opposite phases $e^{\pm i\delta(r,t)}$, where $\delta(r,t)$ is the nonlinear eikonal phase. Their contribution to the symmetric commutator is therefore schematically of the form $1 - \cos \delta(r,t)$. For $|\delta|\ll1$, this begins at order $\delta^2$, which explains the absence of the linear maximal-chaos contribution. The phase-only form $1-\cos\delta$ describes the onset of the nonlinear regime, but not its end point.  Taken literally, it oscillates and never saturates.  We assume that the resummed strip-edge correlators decay in modulus, $|F_\pm|\to0$, as motivated by the zero-dimensional results~\cite{Kitaev:2017awl,MaldacenaStanford2016}. This assumed decay, together with factorisation of the two time-ordered terms and a subleading connected subtraction, produces the conditional plateau in Eq.~\eqref{eq:SYK-plateau-main}. The role of the eikonal phase is therefore twofold.  It explains the absence of a linear maximal-chaos term, and, once $|\delta|=O(1)$, it locates the front behind which the commutator becomes order one.

Near strong coupling, this nonlinear sector gives the butterfly front found in Refs.~\cite{Gu:2018jsv,CHMS2023,ChoiMezeiSarosi2021}. We denote by $u_B$ the front velocity measured in lattice sites per unit dimensionless time $\tau$.  It is
\begin{equation}
u_B^{-1}
=
\operatorname{arccosh}\!\left[
\frac{
1 + \nu + (\eta - 2)\nu^2
}{
\eta \nu^2
}
\right].
\label{eq:SYK-uB-main}
\end{equation}
Since $\tau=2\pi t/\beta$, the corresponding velocity in physical time is
\begin{equation}
v_B
=
\frac{2\pi}{\beta} u_B.
\label{eq:SYK-vB-main}
\end{equation}

To obtain an explicit late-time coefficient, one additional assumption is required. We assume that the full strong-coupling eikonal resummation produces a region saturated to the pair value in Eq.~\eqref{eq:SYK-plateau-main} throughout
\begin{equation}
|r|
<
u_B(\tau-\tau_B),
\end{equation}
up to a transition region whose width is $o(\tau)$, and that the total contribution from outside the front is also $o(\tau)$. We further take $L\to\infty$ with $v_Bt\ll L$, before finite-size wraparound of the front around the chain.  Under these assumptions, the length of the saturated interval is
\begin{equation}
2 u_B(\tau-\tau_B)
=
2 v_B(t-t_B),
\end{equation}
and Eq.~\eqref{eq:SYK-exact-equality} gives
\begin{align}
\mathbb E_J D^{(L)}(t) &=
\mathbb E_J \bar C^{(L)}(t) \nonumber\\
&=
4
\left[
\cos\!\left( \frac{\pi \nu}{2} \right)
\right]^{4\Delta}
v_B (t-t_B)  + o(t).
\label{eq:SYK-late-main}
\end{align}
Thus, the late-time linear growth is simply the saturated pair value multiplied by the length of the expanding ballistic interval.

Finally, the strong-coupling limit must be taken in the correct order. Defining $\epsilon_\nu   = 1-\nu$, the strong coupling corresponds to $0<\epsilon_\nu\ll1$.  We first keep $\epsilon_\nu$ fixed while taking the large-$N$, late-time eikonal limit, and only afterwards expand in $\epsilon_\nu$.  Setting $\nu=1$ from the outset is singular. In this limit, the pole and saddle approximations are nonuniform, while the fixed-$\Delta$ plateau in Eq.~\eqref{eq:SYK-plateau-main} vanishes. Equation~\eqref{eq:SYK-late-main} should therefore be understood as a leading strong-coupling eikonal prediction. Establishing the coefficient to all orders would require uniform control of the multi-string expansion and of the front and exterior regions, as discussed in Appendix~\ref{app:SYK}.

\section{Discussion}\label{sec:discussion}

We have considered three closely related scrambling observables.  The local squared commutator $C_r(t)$ measures the signal at a fixed separation. Its spatial sum $\bar C(t) = \sum_r C_r(t)$ measures the spatial extent of this signal and becomes proportional to the scrambled volume, once an order-one plateau has formed and the exterior weight is controlled.  The dOTOC $D(t)$ contains the same pair contributions together with their zero-momentum cross terms.

The main new results of the paper are the two bounds
\begin{equation}
\bar C(t) \lesssim t^d,
\qquad
D(t) \lesssim t^d.
\label{eq:discussion-bounds}
\end{equation}
The first follows from Lieb--Robinson locality. The second additionally uses equilibrium clustering. The causal-block proof shows where the improvement over the earlier $t^{3d}$ estimate comes from: the thermal fourth moment of a whole source block remains of order the square of its volume. This removes the extra causal-volume factor that appears in a sitewise estimate.

We have also followed the same observables through the scrambling time. A velocity-dependent Lyapunov form for $C_r(t)$ retains the same maximal exponential growth rate after spatial integration. Once the local profile has saturated, a ballistic order-one region with the exterior control of Eq.~\eqref{eq:app-exterior} gives $\bar C(t) = \Theta((t-t_*)^d)$. Positive spatial averages of the standard MSS four-point function obey the same Schwarz--Pick bound and lead to the spatial Legendre constraint in Eq.~\eqref{eq:Legendre-MSS-main}.

The Brownian coupled-cluster model gives the cleanest explicit deterministic crossover. The ensemble symmetry gives $\E_J D = \E_J \bar C$ exactly, and the deterministic approximation crosses over from $8 e^{\gamma t}/(3N)$ to $4v_Bt$ while the finite-$N$ late-time slope is $4v_B(N)$. In the finite-temperature SYK chain, the same exact disorder identity holds for the parity-even probes used here. Its early exponent is $2\pi \nu/\beta$, while the symmetric commutator cancels the linear maximal-chaos pole. The late-time coefficient in Eq.~\eqref{eq:SYK-late-main} relies on the additional all-orders strong-coupling eikonal assumption stated in Appendix~\ref{app:SYK}.

These results give a concrete relation between the usual pre-scrambling Lyapunov regime and the power-law growth discussed in Ref.~\cite{Kukuljan:2017xag}. Once the local amplitude has saturated, further growth can come from the increasing spatial volume over which the commutator is order one.  The detailed crossover, including its coefficient and front structure, remains model dependent. We stress that, in both examples, the cross terms in Eq.~\eqref{eq:D-C-cross} vanish only after disorder averaging.  Thus, $\E_JD=\E_J\bar C$ holds by symmetry, not dynamics, and the examples do not test the mechanism that distinguishes $D$ from $\bar C$ in Theorem~\ref{thm:one-volume}. A natural next step would therefore  be to study a model in which $D\neq\bar C$, e.g., a translation-invariant model without disorder, or a single Brownian realisation, and to ask whether the zero mode can approach the $t^d$ bound with a coefficient different from that of $\bar C$.

\begin{acknowledgments}
The work of S.G.\ is supported by the STFC Ernest Rutherford Fellowship ST/T00388X/1, while T.P.\ is 
supported by ERC Advanced grant
 No.~101096208 -- QUEST.
The work is also supported by the research programme P1-0402 (S.G. and T.P.) and the projects J7-60121 (S.G.) and N1-0368 (T.P.) of the Slovenian Research and Innovation Agency (ARIS). 

We acknowledge the use of the large language models ChatGPT (OpenAI) and Claude (Anthropic) for language editing, checking the internal consistency of derivations and references, assisting with numerics, and suggesting improvements to the presentation. All results were derived and verified by the authors.
\end{acknowledgments} 

\appendix
\section{Thermal norm and a block commutator inequality}
\label{app:thermal}

For the density matrix $\rho=e^{-\beta H}/Z$, the symmetric thermal inner product and norm are
\begin{equation}
\langle A,B\rangle_{\beta}
=\Tr\!\left(\rho^{1/2}A^\dagger\rho^{1/2}B\right),
\qquad
\norm{A}_{\beta}^{2}=\langle A,A\rangle_{\beta}.
\label{eq:app-KMS-def}
\end{equation}
To distinguish this thermal norm from ordinary Schatten norms, we use the notation
$\norm{X}_{S_p}\equiv[\Tr (X^\dagger X)^{p/2}]^{1/p}$ for the Schatten $p$-norm, $1\le p<\infty$, while $\norm{X}_\infty$ denotes the operator norm.  Thus, subscripts $S_2,S_4,\ldots$ always refer to Schatten norms, whereas the subscript $\beta$ refers exclusively to the symmetric thermal norm.
Since $\langle\one,A\rangle_\beta=\langle A\rangle_\beta$ and $\norm{\one}_\beta=1$, orthogonal projection away from the identity gives
\begin{equation}
\norm{\cent A}_{\beta}^{2}
=\norm{A}_{\beta}^{2}-\abs{\langle A\rangle_\beta}^{2}
\le\norm{A}_{\beta}^{2}.
\label{eq:app-centering-contracts}
\end{equation}
Moreover,
\begin{equation}
\norm{A}_{\beta}
=\norm{\rho^{1/4}A\rho^{1/4}}_{S_2}
\le\norm{A}_\infty,
\label{eq:app-KMS-op}
\end{equation}
by the Schatten--H\"older inequality and $\Tr\rho=1$.

\begin{lemma}[Thermal $L_4$ commutator bound]
\label{lem:L4}
For Hermitian operators $A,B$,
\begin{equation}
\norm{\cent{[A,B]}}_\beta^2
\le4\sqrt{\langle A^4\rangle_\beta\langle B^4\rangle_\beta}.
\label{eq:app-L4}
\end{equation}
\end{lemma}

\begin{proof}
The Schatten--H\"older inequality gives
\begin{align}
\norm{AB}_\beta
&=\norm{\rho^{1/4}AB\rho^{1/4}}_{S_2}
\nonumber\\
&\le\norm{\rho^{1/4}A}_{S_4}\norm{B\rho^{1/4}}_{S_4}.
\label{eq:app-L4-holder}
\end{align}
For the first factor,
\begin{align}
&\norm{\rho^{1/4}A}_{S_4}^4
=\Tr\!\left(\rho^{1/2}A^2\rho^{1/2}A^2\right)
\nonumber\\
&\le
\left[\Tr\!\left((\rho^{1/2}A^2)^\dagger
\rho^{1/2}A^2\right)\right]^{1/2}
\left[\Tr\!\left(\rho^{1/2}A^2
(\rho^{1/2}A^2)^\dagger\right)\right]^{1/2}
\nonumber\\
&=\Tr(\rho A^4)=\langle A^4\rangle_\beta.
\label{eq:app-L4-factor}
\end{align}
This is the Hilbert--Schmidt Cauchy--Schwarz inequality with $X=\rho^{1/2}A^2$, namely,   $|\Tr X^2|\le[\Tr(X^\dagger X)\Tr(XX^\dagger)]^{1/2}$. The same estimate holds for $\norm{B\rho^{1/4}}_{S_4}$, and hence
\begin{equation}
\norm{AB}_\beta
\le[\langle A^4\rangle_\beta\langle B^4\rangle_\beta]^{1/4}.
\end{equation}
The same bound holds for $BA$. The use of triangle inequality and Eq.~\eqref{eq:app-centering-contracts} then proves Eq.~\eqref{eq:app-L4}.
\end{proof}

The point of Lemma~\ref{lem:L4} is that $A$ and $B$ can themselves be extensive block observables. For a block containing $O(R^d)$ local terms, clustering will give $\langle A^4\rangle_\beta=O(R^{2d})$. Equation~\eqref{eq:app-L4} then keeps the squared block commutator at the same $O(R^{2d})$ scale, and no operator-norm bound on the full block is necessary. If a commutator is anti-Hermitian, multiplication by $i$ makes it Hermitian without changing its symmetric thermal norm or covariance.

\section{Consequences of ordinary Gibbs clustering}
\label{app:clustering}

We assume Eq.~\eqref{eq:Gibbs-clustering-main} throughout this appendix and keep all constants uniform in volume.

\subsection{\texorpdfstring{Clustering at the centre of the thermal strip}{Clustering at the centre of the thermal strip}}

\begin{lemma}[Symmetric-thermal clustering]
\label{lem:KMS-clustering}
For centred local Hermitian operators $a_X,b_Y$, there are constants $K,c>0$, an integer $m$, and a fixed polynomial $P_{\rm KMS}(X,Y)$ in the support sizes and diameters, all uniform in volume, such that
\begin{align}
\abs{\langle a_X,  b_Y\rangle_\beta}
&\le K P_{\rm KMS}(X,Y) (1+r)^m e^{-cr}
\nonumber\\[-0.2em]
&\quad \times\norm{a_X}_\infty\norm{b_Y}_\infty,
\qquad r= \dist(X,Y).
\label{eq:app-KMS-clustering}
\end{align}
For fixed-support operators $P_{\rm KMS}$ is an $O(1)$ constant.
\end{lemma}

\begin{proof}
For $0\le\Re z\le 1$, define
\begin{equation}
G(z)=\Tr\!\left(\rho^za_X\rho^{1-z}b_Y\right).
\label{eq:app-modular-G}
\end{equation}
At the centre, $G(1/2)=\langle a_X, b_Y\rangle_\beta$. On the boundaries,
\begin{align}
G(iu)&=\langle\tau_{\beta u}(b_Y)a_X\rangle_\beta,
\nonumber\\
G(1+iu)&=\langle\tau_{-\beta u}(a_X)b_Y\rangle_\beta.
\label{eq:app-modular-edges}
\end{align}
These are ordinary connected real-time correlators because the state is stationary.  The noncommutative H\"older inequality then gives the uniform strip bound
\begin{equation}
|G(s+iu)|\le\norm{a_X}_\infty\norm{b_Y}_\infty,
\qquad0\le s\le 1.
\label{eq:app-modular-bounded}
\end{equation}

First consider $G(iu)$, and note that the argument for the other boundary is identical. Before the real-time light cones overlap, approximate $\tau_{\beta u}(b_Y)$ by an operator supported in the $\ell$-neighbourhood of $Y$, centre that approximant by subtracting its thermal expectation, and take
\begin{equation}
\ell = \frac{1}{2} (r + v_{\rm LR} \beta|u|).
\end{equation}
The Lieb--Robinson truncation error and the remaining spatial separation are both proportional to $r-v_{\rm LR}\beta|u|$. The polynomial support factor generated by the truncation grows only as a fixed power of $r+|u|$. We keep it in the prefactor below. Then, it can be absorbed into a slightly weaker exponential at the end. Equation~\eqref{eq:Gibbs-clustering-main} therefore yields, for $v_{\rm LR}\beta|u|\le r/2$,
\begin{equation}
|G(iu)| \le K(1+r)^m
 e^{-c (r - v_{\rm LR}\beta|u|)}.
\label{eq:app-edge-before-overlap}
\end{equation}
For larger $|u|$, use Eq.~\eqref{eq:app-modular-bounded}.

For $r\ge 1$, choose
\begin{equation}
a = \frac{c_0(\beta v_{\rm LR})^2}{r},
\qquad
H(z)=e^{a(z-1/2)^2}G(z),
\label{eq:app-a-choice}
\end{equation}
with any fixed $c_0 > 0$. On either boundary,
\begin{equation}
|H(\sigma + iu)|=e^{a/4 - au^2}|G(\sigma+iu)|,
\qquad\sigma=0,1.
\end{equation}
Writing $x=\beta v_{\rm LR}|u|$, the exponent in the pre-overlap region has the form $-cr+cx-c'x^2/r$, which is bounded by $-c''r$ for $x\le r/2$.  For $x\ge r/2$, the Gaussian and the uniform strip bound again give $e^{-c''r}$. Thus,
\begin{equation}
\sup_u |H(\sigma + iu)| \le K(1+r)^m e^{-c''r},
\qquad\sigma=0, 1.
\end{equation}
Hadamard's three-lines theorem gives the same estimate at $z=1/2$. Changing constants covers bounded $r$. Only real-time Lieb--Robinson locality is used in this interpolation argument.
\end{proof}

At $\beta=0$, the symmetric thermal product is tracial (invariant under cyclic permutations) and this interpolation is unnecessary.

\subsection{Fourth moments of extensive local sums}

\begin{lemma}[Extensive fourth moment]
\label{lem:fourth-moment}
Let $a_x$ be a uniformly bounded family of centred Hermitian observables with fixed support radius. For every finite set $S$,
\begin{equation}
A_S = \sum_{x\in S}a_x
\qquad \longrightarrow\qquad
\langle A_S^4\rangle_\beta\le K_\beta|S|^2.
\label{eq:app-fourth-moment}
\end{equation}
\end{lemma}

\begin{proof}
Expand the fourth moment into ordered quadruples,
\begin{equation}
\langle A_S^4\rangle_\beta
=\sum_{x_1, x_2, x_3, x_4 \in S}
\langle a_{x_1}a_{x_2}a_{x_3}a_{x_4}\rangle_\beta.
\label{eq:app-fourth-expand}
\end{equation}
For a fixed quadruple, we need to separate it into two groups whose distance is comparable with the size of the quadruple. Let $D=\diam(x_1,x_2,x_3,x_4)$ and call two support centres close if their distance is smaller than $D/3$. Split the four centres into groups of centres that are connected by chains of close pairs. If all four belonged to a single group, the two most distant centres would be connected by a chain of at most three close pairs, giving $D<D$, a contradiction. Hence, there are at least two groups. Let $I$ be one of them, let $J$ be the union of the others, and let $\ell$ be the distance between $I$ and $J$. By construction, $I$ and $J$ are nonempty and
\begin{equation}
\ell \ge \frac{1}{3} \diam(x_1, x_2, x_3, x_4).
\label{eq:app-MST-diam}
\end{equation}
If $\ell$ is smaller than a fixed multiple of the support radius, the trivial operator-norm bound is absorbed into $Ke^{-c\ell}$. Otherwise, the two groups have disjoint supports, so operators belonging to different groups commute and the ordered product can be regrouped as $a_I b_J$.  The clustering assumption~\eqref{eq:Gibbs-clustering-main} was stated for arbitrary local operators, so it applies directly to the products $a_I$ and $b_J$.  If one starts from a Hermitian-only formulation, the same statement follows by decomposing each product into its Hermitian and anti-Hermitian parts.  The support polynomial is at most polynomial in the diameter of the quadruple and can be absorbed into $e^{-c'\ell}$ with $c' < c$.  Thus, after changing constants,
\begin{equation}
\abs{\langle a_Ib_J\rangle_\beta
-\langle a_I\rangle_\beta \langle b_J\rangle_\beta}
\le K e^{-c\ell}.
\label{eq:app-fourth-cut}
\end{equation}
There are two possibilities. If one operator is separated from the other three, then the factorised term vanishes because each $a_x$ is centred. If instead the four operators separate into two pairs, the factorised term is a product of two connected two-point functions, each of which is exponentially small by clustering. Therefore, after enlarging constants,
\begin{align}
&\abs{\langle a_{x_1} a_{x_2} a_{x_3}a_{x_4}\rangle_\beta}
\nonumber\\
&\quad\le
Ke^{-c \diam(x_1,x_2,x_3,x_4)}
\nonumber\\[-0.2em]
&\qquad
+K\!\sum_{\text{three pairings}}
 e^{-c \dist(x_i,x_j)}e^{-c\dist(x_k,x_l)}.
\label{eq:app-fourth-pointwise}
\end{align}
For the first term, fix $x_1$. Polynomial lattice growth implies that the number of triples with diameter in $[r,r+1)$ is polynomial in $r$, so the exponential shell sum is finite; its total contribution is $O(|S|)$. For each pairing,
\begin{equation}
\sum_{x, y \in S}e^{- c \dist(x,y)} \le K|S|,
\end{equation}
and the corresponding fourfold sum is therefore $O(|S|^2)$. Summing Eq.~\eqref{eq:app-fourth-pointwise} proves the lemma.
\end{proof}

This direct fourth-moment proof is all that we need below. General results showing that two-point clustering propagates to higher connected correlations in quantum lattice systems are given in Ref.~\cite{AmpelogiannisDoyon2024}.

As a simpler corollary of Lemma~\ref{lem:KMS-clustering}, any centred extensive local sum also has an extensive symmetric-thermal variance:
\begin{equation}
\norm{A_S}_\beta^2 \le K_\beta|S|.
\label{eq:app-second-moment}
\end{equation}

\section{Proof of the one-volume bounds}
\label{app:one-volume}

We begin by replacing $w_x,v_y$ by their centred versions, which  does not change any commutator. We write
\begin{equation}
W_B = \sum_{x\in B}w_x,
\qquad
V_S = \sum_{y\in S}v_y.
\end{equation}
The densities $w_x$ and $v_y$ have supports of a fixed finite radius around their sites. To keep the notation concise, all balls and enlargements below are understood to include this fixed radius, which does not affect any of the estimates.

\subsection{Pair and one-source estimates}

For every pair,
\begin{equation}
\norm{\cent{[w_x(t), v_y]}}_\beta
\le\norm{[w_x(t), v_y]}_\infty.
\end{equation}
The Lieb--Robinson shell sum with $R_{\rm LR}$ from Eq.~\eqref{eq:RLR-main} gives
\begin{align}
\sum_y \norm{\cent{[w_x(t), v_y]}}_\beta^2
&\le K\sum_{r\le R_{\rm LR}}(1+r)^{d-1}
\nonumber\\
&\quad+K\sum_{r>R_{\rm LR}}(1+r)^{d-1}e^{-2\mu_{\rm LR}(r-R_{\rm LR})}
\nonumber\\
&\le K R_{\rm LR}^d.
\label{eq:app-pair-shell}
\end{align}
Averaging over $x$ proves Eq.~\eqref{eq:Csum-bound-main} when $R_{\rm LR}^d \le L$, and the trivial bound $\bar{C}^{(L)} \le KL$  takes care of the complementary finite-size regime.

Now define
\begin{equation}
M_x(t) = \cent{[w_x(t),V]}.
\label{eq:app-Mx}
\end{equation}
Let $S_x$ be a ball of radius $aR_{\rm LR}$ about $x$, with fixed $a>1$ sufficiently large. Summed Lieb--Robinson tails give
\begin{equation}
\norm{[w_x(t), V-V_{S_x}]}_\infty\le Ke^{-c R_{\rm LR}}.
\label{eq:app-source-tail}
\end{equation}
For the local target sum, Lemmas~\ref{lem:L4} and \ref{lem:fourth-moment}, together with stationarity, yield
\begin{align}
\norm{\cent{[w_x(t),V_{S_x}]}}_\beta^2
&\le4\sqrt{\langle w_x(t)^4\rangle_\beta
              \langle V_{S_x}^4\rangle_\beta}
\nonumber\\
&=4\sqrt{\langle w_x^4\rangle_\beta
              \langle V_{S_x}^4\rangle_\beta}
\le K_\beta R_{\rm LR}^d.
\label{eq:app-source-bound}
\end{align}
The tail in Eq.~\eqref{eq:app-source-tail} is negligible, so
\begin{equation}
\norm{M_x(t)}_\beta^2 \le K_\beta R_{\rm LR}^d.
\label{eq:app-Mx-bound}
\end{equation}
For later use, define the source-resolved average
\begin{equation}
S_L(t)=\frac1L\sum_x\norm{M_x(t)}_\beta^2.
\label{eq:app-S-def}
\end{equation}
Averaging gives $S_L(t)\le K_\beta\min\{R_{\rm LR}^d,L\}$, with the global fourth-moment estimate supplying the finite-size cap.

\subsection{Causal source blocks}

Assume first that $R_{\rm LR}^d \le cL$, and partition each periodic direction into intervals of length between $R_{\rm LR}$ and $2R_{\rm LR}$, producing cubes $B$ with comparable side lengths. Define
\begin{equation}
M_B(t) = \sum_{x\in B} M_x(t) = \cent{[W_B(t),V]}.
\label{eq:app-MB}
\end{equation}
Let $S_B$ be an $aR_{\rm LR}$-enlargement of $B$. Then
\begin{equation}
\norm{[W_B(t), V-V_{S_B}]}_\infty
\le K |B| e^{- c R_{\rm LR}}.
\label{eq:app-block-tail}
\end{equation}
For the local part,
\begin{align}
\norm{\cent{[W_B(t), V_{S_B}]}}_\beta^2
&\le 4 \sqrt{\langle W_B(t)^4\rangle_\beta
              \langle V_{S_B}^4\rangle_\beta}
\nonumber\\
&=4 \sqrt{\langle W_B^4 \rangle_\beta
              \langle V_{S_B}^4 \rangle_\beta}
\nonumber\\
&\le K_\beta |B||S_B|
\le K_\beta |B|R_{\rm LR}^d.
\label{eq:app-block-bound}
\end{align}
The tail obeys the same weaker bound after changing the constant. Hence
\begin{equation}
\norm{M_B(t)}_\beta^2\le K_\beta|B|R_{\rm LR}^d.
\label{eq:app-MB-bound}
\end{equation}
This is the step that saves one volume. A causal block contains $|B|=\Theta(R_{\rm LR}^d)$ sources. Estimating its sources one by one would allow another factor of  $R_{\rm LR}^d$. Equation~\eqref{eq:app-MB-bound} keeps the whole block at $O(R_{\rm LR}^{2d})$ before the sum over blocks is performed.

\subsection{Distant source covariances}

Let $r = \dist(x,x') \ge b R_{\rm LR}$, with $b$ a sufficiently large constant. Choose a Hermitian Lieb--Robinson truncation $w_x^{(\ell)}(t)$, supported in the $\ell$-ball about $x$, with $\ell=\lfloor r/4 \rfloor$. Define
\begin{equation}
M_x^{(\ell)}(t) = \cent{[w_x^{(\ell)} (t), V]}.
\end{equation}
Only target densities meeting the truncated support contribute, so $M_x^{(\ell)}$ is supported in a fixed enlargement of that ball and
\begin{equation}
\norm{M_x^{(\ell)}(t)}_\infty \le K(1+\ell)^d.
\label{eq:app-Mx-local-op}
\end{equation}
For completeness, let $T_x$ be a ball of radius $2\ell$ around $x$ and write $V=V_{T_x}+V_{T_x^c}$. For the targets in $T_x$,
\begin{equation}
\norm{[w_x(t)-w_x^{(\ell)}(t), V_{T_x}]}_\beta
\le K(1 + \ell)^d e^{-\mu_{\rm LR}(\ell - v_{\rm LR} |t| - R_0)},
\end{equation}
using the Lieb--Robinson truncation error and the number of local terms in $T_x$. For the targets outside $T_x$, the commutator with $w_x^{(\ell)}(t)$ vanishes after enlarging $T_x$ by the fixed support radius, while the commutator with $w_x(t)$ is bounded by the usual Lieb--Robinson shell sum. Combining the two parts gives
\begin{equation}
\norm{M_x(t)-M_x^{(\ell)}(t)}_\beta
\le K(1+\ell)^d e^{-\mu_{\rm LR}(\ell-v_{\rm LR}|t|-R_0)}.
\label{eq:app-Mx-localise}
\end{equation}
The two truncated commutators centred at $x$ and $x'$ are separated by a fixed fraction of $r$. Applying Lemma~\ref{lem:KMS-clustering} to their Hermitian multiples $iM_x^{(\ell)}$ and $iM_{x'}^{(\ell)}$, and using Eqs.~\eqref{eq:app-Mx-bound}--\eqref{eq:app-Mx-localise} for the errors, gives
\begin{equation}
\abs{\langle M_x(t),M_{x'}(t)\rangle_\beta}
\le K_\beta(1+r)^m e^{-cr},
\qquad r\ge bR_{\rm LR}.
\label{eq:app-far-source}
\end{equation}
All powers of $R_{\rm LR}$ are absorbed into the polynomial because $R_{\rm LR}\le r/b$.

\subsection{The global block sum}

Since $\sum_BM_B=\cent{[W(t),V]}$,
\begin{equation}
LD^{(L)}(t)
=\sum_{B,B'}\langle M_B(t),M_{B'}(t)\rangle_\beta.
\label{eq:app-global-block-sum}
\end{equation}
Call two blocks near if their distance is at most $bR_{\rm LR}$. Each block has only $O(1)$ near neighbours. The Cauchy--Schwarz inequality and Eq.~\eqref{eq:app-MB-bound} therefore give
\begin{align}
\sum_{B,B'\,\mathrm{near}}
\abs{\langle M_B,M_{B'}\rangle_\beta}
&\le K\sum_B\norm{M_B}_\beta^2
\nonumber\\
&\le K_\beta R_{\rm LR}^d\sum_B|B|
=K_\beta LR_{\rm LR}^d.
\label{eq:app-near-blocks}
\end{align}
For far blocks, expand back into sources and use Eq.~\eqref{eq:app-far-source}:
\begin{align}
\sum_{B,B'\,\mathrm{far}}
\abs{\langle M_B,M_{B'}\rangle_\beta}
&\le\sum_x
\sum_{\substack{x':\,\dist(x,x')\\ \ge bR_{\rm LR}}}
 K_\beta(1+r)^m e^{-cr}
\nonumber\\
&\le K_\beta L.
\label{eq:app-far-blocks}
\end{align}
Equations~\eqref{eq:app-global-block-sum}--\eqref{eq:app-far-blocks} prove $D^{(L)}(t)\le K_\beta R_{\rm LR}^d$ when $R_{\rm LR}^d\lesssim L$.

If $R_{\rm LR}^d\gtrsim L$, use the entire lattice as one block. Lemmas~\ref{lem:L4} and \ref{lem:fourth-moment} give
\begin{align}
LD^{(L)}(t)
&\le4\sqrt{\langle W(t)^4  \rangle_\beta
             \langle V^4\rangle_\beta}
\nonumber\\
&=4 \sqrt{\langle W^4\rangle_\beta
           \langle V^4\rangle_\beta}
\le K_\beta L^2.
\label{eq:app-global-finite-size}
\end{align}
Thus, $D^{(L)}(t)\le K_\beta L\le K_\beta R_{\rm LR}^d$.  This completes the proof of Theorem~\ref{thm:one-volume}.

\subsection{Momentum and random-sign identities}

On a periodic lattice define
\begin{equation}
W_p=\sum_xe^{-ip\cdot x}w_x,
\qquad
V_k=\sum_ye^{-ik\cdot y}v_y,
\end{equation}
and
\begin{equation}
\mathcal P(p,k;t)=\frac1L\norm{\cent{[W_p(t),V_k]}}_\beta^2\ge0.
\label{eq:app-P}
\end{equation}
Double Fourier orthogonality gives
\begin{align}
\bar{C}^{(L)}(t)&=\frac1{L^2}\sum_{p,k}\mathcal P(p,k;t),
\nonumber\\
S_L(t)&=\frac1L\sum_p\mathcal P(p,0;t),
\nonumber\\
D^{(L)}(t)&=\mathcal P(0,0;t).
\label{eq:app-P-rel1}
\end{align}
Thus,  $\bar{C}^{(L)}$ is the two-momentum average, $S_L$ fixes the target momentum to zero, and the dOTOC is the double zero mode.

The same distinction can be made without Fourier space. Let $\epsilon_x$ and $\eta_y$ be independent random signs and set
\begin{equation}
W_\epsilon=\sum_x\epsilon_xw_x,
\qquad
V_\eta=\sum_y\eta_yv_y.
\end{equation}
Then
\begin{equation}
\E_{\epsilon,\eta}\frac1L
\norm{\cent{[W_\epsilon(t),V_\eta]}}_\beta^2
=\bar{C}^{(L)}(t),
\label{eq:app-random-sign-Csum}
\end{equation}
and randomising only the source gives $S_L$. Thus, $\bar C^{(L)}$ is a momentum average, while $D^{(L)}$ singles out the zero mode.

\subsection{Locality and positivity alone are insufficient for Theorem~\ref{thm:one-volume}}

The block-moment constraint is genuinely stronger than positivity and pairwise locality. Let $B_R(x)$ denote the lattice ball of radius $R$ about $x$. Let $\{e_y\}$ be orthonormal and define
\begin{equation}
a_{xy}^{(R)}=e_y\,\mathbf1_{\{y\in B_R(x)\}}.
\label{eq:app-sharp-pair}
\end{equation}
Then, 
\begin{equation}
\frac1L\sum_{x,y}\norm{a_{xy}^{(R)}}^2
=\Theta(R^d),
\end{equation}
but
\begin{equation}
\frac1L\norm{\sum_{x,y}a_{xy}^{(R)}}^2
=\Theta(R^{2d}).
\label{eq:app-sharp-result}
\end{equation}
The construction is purely kinematic. For a causal source block containing $n=\Theta(R^d)$ sources, its coherent block sum has squared norm $\Theta(n^3)$, whereas Eq.~\eqref{eq:app-block-bound} permits only $O(n^2)$. The toy example therefore explains why locality and positivity do not prove the theorem by themselves. The fourth-moment constraint on a causal source block is the additional input that rules out this coherent $R^{2d}$ scaling.

\section{Positive spatial averages and the Schwarz--Pick theorem}
\label{app:spatial-MSS}

Let
\begin{equation}
S_\beta = \{z \in\mathbb C:\ |\Im z|<\beta/4\},
\end{equation}
and suppose each normalised local OTOC $F_r:S_\beta\to\mathbb D$ is analytic. For positive weights $a_r\ge0$, with $a_{\rm tot}=\sum_ra_r>0$, define
\begin{equation}
F_a(z)=\frac1{a_{\rm tot}}\sum_ra_rF_r(z).
\end{equation}
The unit disk is convex, so $|F_a|\le 1$. Infinite summable families follow by uniform convergence.

The conformal map
\begin{equation}
\phi(z)= \tanh\frac{\pi z}{\beta}
\end{equation}
sends $S_\beta$ to the unit disk. The Schwarz--Pick theorem gives, for real $t$,
\begin{equation}
\frac{|\partial_t F_a(t)|}{1- |F_a(t)|^2}
\le\frac{|\phi'(t)|}{1-|\phi(t)|^2}
=\frac{\pi}{\beta}.
\label{eq:app-SP-strip}
\end{equation}
For
\begin{equation}
G_a(t) = \sum_ra_r[1-F_r(t)]=a_{\rm tot}[1-F_a(t)],
\end{equation}
and using $1-|F|^2\le2|1-F|$,
\begin{equation}
|\partial_tG_a(t)| \le\frac{2\pi}{\beta}|G_a(t)|.
\label{eq:app-positive-MSS}
\end{equation}
If $F_a(t)$ is real and the deficit is increasing, one also has
\begin{equation}
\partial_tG_a
\le\frac{2\pi}{\beta}G_a
-\frac{\pi}{\beta a_{\rm tot}} G_a^2.
\end{equation}

For finite volume and $t_2\ge t_1$, integration gives
\begin{equation}
|G_a(t_2)|
\le e^{\frac{2\pi}{\beta}(t_2-t_1)}|G_a(t_1)|.
\label{eq:app-positive-MSS-integrated}
\end{equation}
Now let $\eps_N$ be the small large-$N$ parameter and assume the fixed-time normalised limit
\begin{equation}
g_r(t)=\lim_{N\to\infty}\frac{1-F_r(t)}{\eps_N}
\label{eq:app-linearised-limit}
\end{equation}
exists. For Hermitian probes the MSS-regularised correlator is real on the real axis, and the assumed bound in Eq.~\eqref{eq:MSS-strip-local} gives $1-F_r(t) \ge 0$. Thus, for every $\mu$ for which the sum converges,
\begin{equation}
G_\mu^{\rm lin}(t) =\sum_r e^{\mu\cdot r} g_r(t)
\end{equation}
is positive.  Next, we apply Eq.~\eqref{eq:app-positive-MSS-integrated} to a fixed finite set of sites and divide by $\eps_N$. We first take the thermodynamic and large-$N$ limits for this finite set and then let the set cover the lattice. Since the summands are non-negative, monotone convergence applies.  This yields the following inequality:
\begin{equation}
G_\mu^{\rm lin}(t_2)
\le e^{\frac{2\pi}{\beta}(t_2-t_1)}G_\mu^{\rm lin}(t_1).
\label{eq:app-linearised-integrated}
\end{equation}

Assume now the uniform large-deviation form
\begin{equation}
g_r(t) = \exp\{t\lambda_F(r/t) +o(t)\},
\label{eq:app-linearised-LD}
\end{equation}
with  exponential tightness after each admissible tilt. Then,
\begin{equation}
\lim_{t\to\infty} \frac{1}{t} \ln G_\mu^{\rm lin} (t)
=\sup_v [\lambda_F(v)+\mu\cdot v].
\end{equation}
Applying Eq.~\eqref{eq:app-linearised-integrated} between a fixed $t_1$  and $t$, then taking $t\to\infty$, yields
\begin{equation}
\sup_v [\lambda_F(v) +\mu\cdot v]
\le\frac{2\pi}{\beta}.
\label{eq:app-Legendre-MSS}
\end{equation}
Equivalently,
\begin{equation}
\lambda_F(v) \le \frac{2\pi}{\beta}-\mu\cdot v.
\end{equation}
Optimising over admissible $\mu$ gives the Legendre-transform form of the velocity-dependent bound discussed in the main body of the paper.

\section{Spatial asymptotics at large \texorpdfstring{$N$}{N}}
\label{app:largeN}

\subsection{Before saturation}

Assume that on a bounded velocity domain ${\cal V}$ containing the origin,
\begin{equation}
C_r(t)=\eps_N\exp\{t\lambda(r/t)+o(t)\},
\label{eq:app-LD}
\end{equation}
uniformly, with $C_r\ge0$ and $\lambda(v)\le\lambda(0)=\lambda_L$. Here, the nonzero subexponential prefactor near $v=0$ is absorbed into the $o(t)$ term.  Assume also an exponentially smaller exterior,
\begin{equation}
\sum_{r/t \notin{\cal V}} C_r(t)
\le\eps_N e^{(\lambda_L-c) t+o(t)}.
\label{eq:app-LD-out}
\end{equation}
A single site with $r=O(1)$ gives the lower exponential rate $\lambda_L$. For the upper bound, there are only $O(t^d)$ sites with $r/t\in{\cal V}$, so
\begin{equation}
\sum_{r/t\in{\cal V}}C_r(t)
\le\eps_N t^d e^{\lambda_Lt+o(t)}.
\end{equation}
Together with Eq.~\eqref{eq:app-LD-out}, along any sequence
$t\to\infty$, $N\to\infty$,
\begin{equation}
\frac{1}{t} \ln\frac{\bar{C} (t)}{\eps_{N}}
\to\lambda_L.
\end{equation}
To keep the sequence before local saturation, the full profile in Eq.~\eqref{eq:app-LD}, including its subexponential prefactor, must tend uniformly to zero.

The prefactor follows from the local shape of $\lambda(v)$. If
\begin{equation}
\lambda(v) = \lambda_L-\frac{1}{2} v^{\mathsf T}\mathsf H v+O(|v|^3),
\qquad \mathsf H>0,
\end{equation}
then the saddle has $r=O(\sqrt t)$ and
\begin{equation}
t\lambda(r/t)=\lambda_Lt-\frac1{2t}r^{\mathsf T}\mathsf H r+o(1).
\end{equation}
If the prefactor tends to $b_0\neq0$ on this scale,
\begin{equation}
\bar{C}(t)\sim\eps_Ne^{\lambda_Lt}b_0
\frac{(2\pi t)^{d/2}}{\sqrt{\det\mathsf H}}.
\label{eq:app-Gaussian-Laplace}
\end{equation}
A heat-kernel prefactor $t^{-d/2}$ cancels this saddle volume, as in both explicit models. More generally, $\lambda_L-\lambda(v)\sim c|v|^\alpha$ then gives a saddle volume $t^{d(1-1/\alpha)}$  for $\alpha\ge 1$ without additional time-dependent prefactors.

\subsection{First and deep saturation}

In the OTOC `growth model' stated in Eq.~\eqref{eq:clipped}, local saturation is determined by
\begin{equation}
\eps_Ne^{t\lambda(r/t)}\sim1.
\end{equation}
Writing $\eps_N=e^{-\lambda_Lt_*}$ and $t=t_*+\Delta t$, with $\Delta t\ll t_*$, gives
\begin{equation}
r_{\rm sat}\sim t_*^{1-1/\alpha}(\Delta t)^{1/\alpha}
\label{eq:app-rsat-alpha}
\end{equation}
when $\lambda_L-\lambda(v)\sim c|v|^\alpha$. For a smooth quadratic maximum, $r_{\rm sat}\sim\sqrt{t_*\Delta t}$.  Thus, the first saturated region is not generically a rigid front of width proportional to $\Delta t$.

Deep after scrambling, much less information is required. Suppose there are constants $a_->0$, $a_+<\infty$ and velocities $0<v_-<v_+$ such that $C_r(t)\ge a_-$ for $|r|\le v_-\Delta t$ and $C_r(t)\le a_+$ for $|r|\le v_+\Delta t$, and suppose further that
\begin{equation}
\sum_{|r|>v_+\Delta t}C_r(t)=o((\Delta t)^d).
\label{eq:app-exterior}
\end{equation}
Positivity then gives
\begin{equation}
a_-|B_{v_-\Delta t}|
\le \bar{C}(t)
\le a_+|B_{v_+\Delta t}|+o((\Delta t)^d).
\end{equation}
Since $|B_{v\Delta t}|=\Theta((\Delta t)^d)$, we obtain $\bar{C}(t)=\Theta((\Delta t)^d)$.

\section{Brownian coupled clusters}
\label{app:BCC}

We use the Brownian coupled-cluster dynamics of Xu and Swingle~\cite{XuSwingle2019}. We quote their closed operator-weight equation and derive only the additional relation needed here: the full dOTOC equals the pair sum after Brownian averaging.

It is enough to show that every cross term between two distinct source--target pairs changes sign under a symmetry of the Brownian measure. For $w_x=\sigma^z_{x,1}$ and $v_x=\sigma^z_{x,2}$, conjugate an entire Brownian history by a fixed Pauli matrix on source spin $(x,1)$. The Brownian measure is invariant, while every commutator $[w_x(t), v_y]$ changes sign and commutators with a different source do not. If two distinct pairs share the same source, conjugating the appropriate target spin flips only one of them. Hence,  every off-diagonal pair inner product is odd under an exact measure-preserving transformation:
\begin{equation}
\E_J \left\langle [w_x(t),v_y],[w_{x'}(t), v_{y'}] \right\rangle_0=0,
\qquad (x,y)  \neq(x',y').
\label{eq:app-BCC-exact}
\end{equation}
Since the infinite-temperature trace of each commutator vanishes, Eq.~\eqref{eq:app-BCC-exact} gives Eq.~\eqref{eq:BCC-equality} exactly for every finite $L,N,t$.

For completeness, we state the Xu--Swingle asymptotic inputs used in the main text and the numerical procedure used for Fig.~\ref{fig:BCC-crossover}.  Here, $g$ is their dimensionless nearest-neighbour Brownian coupling.  Their closed cluster-weight description applies after a short microscopic internal-relaxation transient $t_{\rm iso} = O(1)$. In that scaling regime, the linearised cluster weight has total mass $N^{-1}e^{\gamma t}$, with $\gamma=3(1+g^2)$, and their probe-averaged squared commutator is $8/3$ times that weight.  After disorder averaging, permutation symmetry within a cluster makes the fixed target spin used here equivalent to the probe average away from the source cluster; the source-cluster mismatch is $O(N^{-1})$ and does not affect the asymptotic formulas below.  This gives
\begin{equation}
\E_J\bar{C}^{(L)}(t) = \frac{8}{3N} e^{\gamma t}[1+o(1)]
\label{eq:app-BCC-early}
\end{equation}
in the correlated scaling regime $t_{\rm iso}\ll t$, $t\to\infty$, $N\to\infty$, $N^{-1}e^{\gamma t}\to0$. The selected lattice front follows from the linearised discrete equation and has
\begin{equation}
v_B =\min_{\kappa > 0} \frac{3 [1+g^2\cosh \kappa]}{\kappa}.
\label{eq:app-BCC-vB}
\end{equation}
Behind that front the local squared commutator tends to $2$. Summing over the interval $|r|< v_B(N) t$ therefore gives
\begin{equation}
\E_J \bar{C}^{(L)} (t)= 4 v_B(N) t + o(t)
\label{eq:app-BCC-late}
\end{equation}
in the disorder-averaged large-$N$, `pre-wraparound' late-time regime. Combining Eqs.~\eqref{eq:app-BCC-exact}, \eqref{eq:app-BCC-early}, and \eqref{eq:app-BCC-late} yields the early and late asymptotics  quoted in Eqs.~\eqref{eq:BCC-early-main} and \eqref{eq:BCC-crossover}. 

\subsection{Numerical integration of the large-\texorpdfstring{$N$}{N} lattice equation}

The drift of the Brownian weight master equation gives Eq.~\eqref{eq:BCC-mean-field}. Its linearisation around $\phi=0$ is
\begin{equation}
\dot\phi_r
= 3\phi_r+\frac{3g^2}{2} (\phi_{r-1}+\phi_{r+1}).
\label{eq:app-BCC-linear-lattice}
\end{equation}
For a mode $\phi_r\propto e^{-\kappa(r-vt)}$, this yields
\begin{equation}
v(\kappa)=\frac{3[1+g^2\cosh\kappa]}{\kappa},
\label{eq:app-BCC-dispersion}
\end{equation}
whose minimum is precisely $v_B $. Summing Eq.~\eqref{eq:app-BCC-linear-lattice} over $r$ gives $\dot{\bar{C}}_{\rm mf}=\gamma\bar{C}_{\rm mf}$, with $\gamma=3(1+g^2)$. The nonlinear factor in Eq.~\eqref{eq:BCC-mean-field} bounds $0 \le \phi_r\le3/4$ and converts the growing seed into a translating saturated interval.

For Fig.~\ref{fig:BCC-crossover}, we solve Eq.~\eqref{eq:BCC-mean-field} with an adaptive Runge--Kutta 5(4) method on a periodic chain of $801$ sites, using relative tolerance $10^{-9}$, absolute tolerance $10^{-14}$, and final time $t = 18$.  The front remains more than three hundred sites from the periodic boundary, so wraparound is negligible. The late comparison fixes the slope to $4v_B$ and fits only the effective delay $t_0$, giving $t_0=4.000$ over the displayed fit window; the quoted free-slope fit provides an independent numerical check of the asymptotic coefficient of the deterministic equation.  As for a pulled front, the front position also contains the Bramson shift  $-\tfrac{3}{2\kappa_*}\ln t + O(1)$, so a constant delay extracted over a finite window is necessarily window dependent.

\section{Static finite-temperature SYK chain}
\label{app:SYK}

This appendix contains statements at three different levels. The disorder identity $\E_J  D =\E_J \bar C$ is exact at finite $N,q,L,\beta$ for $N>2s$. The Lyapunov profile is a leading-large-$q$ effective-theory result for $N\gg q^2\gg1$, with the time-uniformity qualification stated in Section~\ref{sec:SYK-main}. The thermal plateau additionally requires the pointwise dynamical assumptions in Appendix~\ref{app:SYK-plateau}. The explicit late linear coefficient additionally assumes uniform all-orders control of the strong-coupling eikonal expansion. We keep these three statements separate below. The notation $\E_J$ denotes the average over the static Gaussian coupling tensors after the thermal trace has been taken for the fixed Hamiltonian $H_J$. We also keep separate the extrapolated pole-free scale $\tau_{\rm pf}$ and the nonlinear energy-momentum tensor front delay $\tau_B$, as they have different large-$N$ scalings at fixed $0 < \nu < 1$.

\subsection{Parity-even probes and the exact disorder identity}

First, write
\begin{align}
\frac{\pi\nu}{\cos(\pi\nu/2)}&=\beta\mathcal J,
\qquad
\mathcal J^2=\mathcal J_0^2+\mathcal J_1^2,
\nonumber\\
\eta&=\frac{\mathcal J_1^2}{\mathcal J^2},
\qquad 0< \eta \le 1.
\label{eq:app-SYK-parameters}
\end{align}
Let $A$ and $B$ be disjoint flavour sets with
\begin{equation}
|A|=|B|=s\in2\mathbb Z,
\qquad N>2s,
\qquad \frac{s}{q}\to\Delta>0.
\end{equation}
Define normalised Hermitian strings
\begin{equation}
o_{S,x}=2^{|S|/2}i^{|S|(|S|-1)/2}\prod_{a\in S}\chi_{x,a},
\qquad o_{S,x}^2=1.
\end{equation}
We take $w_x=o_{A,x}$ and $v_x=o_{B,x}$. Both probes are fermion-parity even.

For a fixed disorder realisation, we now define the centred pair commutator
\begin{equation}
M_{xy}^{(J)}(t)
=[w_x(t),v_y]-\langle[w_x(t),v_y]\rangle_{\beta,J}\one.
\label{eq:app-SYK-centred-pair}
\end{equation}
The isotropic Gaussian disorder law is invariant under independent local $SO(N)$ flavour rotations. Suppose first that $x\neq x'$. Choose one flavour in $A$ and one unused flavour outside $A\cup B$, and rotate by $\pi$ in that two-plane at site $x$. The transformation flips $w_x$, leaves every $v_y$ and $w_{x'}$ unchanged, and maps $J$ to an equally probable realisation $J^R$. Since
\begin{equation}
H_{J^R}=U_RH_JU_R^\dagger,
\qquad
\rho_{J^R}=U_R\rho_JU_R^\dagger,
\end{equation}
the centred operator transforms as
\begin{equation}
M_{xy}^{(J^R)}(t)
=-U_R   M_{xy}^{(J)}(t)   U_R^\dagger.
\end{equation}
The other pair transforms by the same conjugation but without a sign. If instead $x=x'$ and $y\neq y'$, the same argument uses a $\pi$ rotation of one flavour in $B$ and an unused flavour at site $y$. Therefore,
\begin{equation}
\E_J \langle M_{xy}^{(J)}(t),M_{x'y'}^{(J)}(t)\rangle_{\beta,J} = 0,
\qquad (x,y) \neq (x',y').
\label{eq:app-SYK-cross-zero}
\end{equation}
Summing Eq.~\eqref{eq:app-SYK-cross-zero} proves Eq.~\eqref{eq:SYK-exact-equality} exactly. The proof uses centring explicitly and therefore does not require an odd probe or a vanishing pair expectation in each realisation. 

\subsection{Pole-free Lyapunov mode}

Let $\kappa(p)$ denote the dimensionless growth exponent of the
momentum-$p$ scrambling mode, so that its real-time dependence is
$e^{\kappa(p)\tau}$, with $\tau=2\pi t/\beta$. The corresponding
physical growth rate is $\lambda(p)= (2\pi/\beta) \kappa(p)$.
Ref.~\cite{CHMS2023} finds
\begin{equation}
\kappa(p)
=
\frac{\nu}{2}
\left[
\sqrt{9-4\eta(1-\cos p)}-1
\right]
=
\nu-\frac{\nu\eta}{6}p^2+O(p^4).
\label{eq:app-SYK-kappa}
\end{equation}
In the large-$q$ eikonal approximation, an external string containing $O(q)$ fermions has a connected vertex of order $q^2/N$~\cite{CHMS2023}. After the factor $2\cos(\pi\kappa/2)$ cancels the one-string energy-momentum tensor pole, the remaining symmetric-contour vertex is finite at $p=0$. We denote its value, together with the normalisation of the external strings, by $\mathcal A_{\Delta,\nu,\eta}$. For $|r|=O(\sqrt{\alpha_{\rm SYK} t })$, $\alpha_{\rm SYK} t\gg1$ and $\sqrt{\alpha_{\rm SYK}t} \ll L$, the small-momentum Fourier saddle gives
\begin{align}
\E_JC_r(t)
&=\mathcal A_{\Delta,\nu,\eta}\frac{q^2}{N}
\frac{e^{\lambda_L t-r^2/(4\alpha_{\rm SYK}t)}}
{\sqrt{4\pi\alpha_{\rm SYK}t}}
[1+o(1)],
\nonumber\\
\alpha_{\rm SYK}
&=\frac{\pi\nu\eta}{3\beta}.
\label{eq:app-SYK-Gaussian}
\end{align}
with $\lambda_L=2\pi\nu/\beta$. Here, $\mathcal A_{\Delta, \nu,\eta}\ge 0$ collects the normalised external vertex and the finite symmetric-edge remainder.  For probes with nonzero overlap with the pole-free mode it is strictly positive and smooth at fixed nonsingular $0<\nu<1$; if the overlap vanishes, the displayed leading term is absent.

Equation~\eqref{eq:app-SYK-Gaussian} is a pre-scrambling asymptotic. Summing over $r$ removes the heat-kernel normalisation and, using Eq.~\eqref{eq:SYK-exact-equality}, gives Eq.~\eqref{eq:SYK-early}. At $r=0$, extrapolating the pole-free result by itself, the extra $\tau^{-1/2}$ factor defines the scale
\begin{equation}
\nu \tau_{\rm pf}
=\ln\frac{N}{q^2}+ \frac{1}{2} \ln \tau_{\rm pf}+O(1).
\label{eq:app-SYK-Tpf}
\end{equation}
At strict $N=\infty$ and fixed $t$, the pair commutator remains infinitesimal; a nontrivial late-time limit must keep $N$ large but finite while time grows logarithmically with $N$.

\subsection{Saturated thermal pair value}
\label{app:SYK-plateau}

With the normalisation $\chi_a^2=1/2$, the large-$q$ Euclidean propagator is~\cite{MaldacenaStanford2016}
\begin{align}
G(\tau_E)
&=\frac{1}{2} \operatorname{sgn}(\tau_E)
\exp\!\left[\frac{g_\beta (\tau_E)}{q}+O(q^{-2})\right],
\nonumber\\
e^{g_\beta(\tau_E)}
&=\left[
\frac{\cos(\pi\nu/2)}
{\cos\!\left(\pi\nu( \frac{1}{2} - \frac{|\tau_E|}{\beta})\right)}
\right]^2.
\label{eq:app-SYK-two-point}
\end{align}
At half the thermal circle,
\begin{equation}
2G(\beta/2)=
\left[\cos\!\left(\frac{\pi\nu}{2}\right)\right]^{2/q}
[1+O(q^{-2})].
\label{eq:app-SYK-half-circle}
\end{equation}
For a string of distinct flavours, large-$N$ factorisation gives
\begin{align}
\norm{o_{S}}_\beta^2
&=\left\langle o_{S}(-i\beta/2)o_{S}(0)\right\rangle_\beta
\nonumber\\
&=[2G(\beta/2)]^{|S|}
\left[1+O(q^{-1})+O(q^2/N)\right]
\nonumber\\
&=\left[\cos\!\left(\frac{\pi\nu}{2}\right)\right]^{2\Delta_S}
[1+o(1)],
\qquad \Delta_S=\frac{|S|}{q}.
\label{eq:app-SYK-string-norm}
\end{align}

Expand the symmetric pair commutator into two time-ordered pieces and two OTOCs at the MSS strip edges.  Assume that, locally deep in a scrambled region, the resummed strip-edge OTOCs vanish, the two time-ordered pieces each factorise into $\norm{w}_\beta^2\norm{v}_\beta^2$, and the connected one-point subtraction in Eq.~\eqref{eq:app-SYK-centred-pair} is subleading. Under these assumptions,
\begin{equation}
C_r(t)\to
2\norm{w}_\beta^2\norm{v}_\beta^2+o(1),
\label{eq:app-SYK-factorised-plateau}
\end{equation}
and, since $\Delta_A=\Delta_B=\Delta$,
\begin{equation}
C_\infty
=2\left[\cos\!\left(\frac{\pi\nu}{2}\right)\right]^{4\Delta}[1+o(1)].
\label{eq:app-SYK-Cinf}
\end{equation}
This is a conditional  leading large-$N$, large-$q$ plateau statement, unlike the exact disorder identity in Eq.~\eqref{eq:SYK-exact-equality}. If $\epsilon_\nu=1-\nu\ll1$,
\begin{equation}
C_\infty
=2\left(\frac{\pi\epsilon_\nu}{2}\right)^{4\Delta}
[1+O(\epsilon_\nu)+o(1)].
\label{eq:app-SYK-Cinf-strong}
\end{equation}

\subsection{Linear pole cancellation and nonlinear energy-momentum tensor front}

Let $\mathcal F(z)=\Tr[\rho^{1/4} v \rho^{1/4} w(z) \rho^{1/4} v \rho^{1/4} w(z)]$, where $z$ denotes complex time. Expanding the symmetric commutator gives two time-ordered terms and the two OTOC values at $t\pm i\beta/4$. A one-string mode $e^{\kappa\tau}$ therefore enters the connected symmetric commutator multiplied by
\begin{equation}
2\cos\!\left(\frac{\pi\kappa}{2}\right).
\label{eq:app-SYK-edge-cos}
\end{equation}
The leading spatial OTOC contains $1/\cos(\pi\kappa/2)$ \cite{CHMS2023}, so the $\kappa=1$ energy-momentum tensor pole cancels at linear order. This is the SYK-chain realisation of the suppression of the maximal contribution to commutator squares discussed in Ref.~\cite{BlakeLiu2021}.

The cancellation concerns the term linear in the energy-momentum tensor exchange. In the coherent energy-momentum tensor sector of the leading eikonal approximation, the two strip edges carry opposite phases,
\begin{equation}
\mathcal F_\pm\simeq
\norm{w}_\beta^2\norm{v}_\beta^2 e^{\pm i\delta_{\rm eik}(r,t)},
\end{equation}
so their contribution is schematically
\begin{equation}
C_{\rm EM}(r,t)\simeq
2\norm{w}_\beta^2\norm{v}_\beta^2
[1-\cos\delta_{\rm eik}(r,t)].
\label{eq:app-SYK-eikonal-comm}
\end{equation}
For small $\delta_{\rm eik}$, Eq.~\eqref{eq:app-SYK-eikonal-comm} starts at order $\delta_{\rm eik}^2$. A complete resummation can also change the modulus of the strip-edge correlators, so Eq.~\eqref{eq:app-SYK-eikonal-comm} is only a schematic description of the transition region. When $|\delta_{\rm eik}|=O(1)$ the energy-momentum tensor sector can nevertheless become order one and set the saturation front.

Continuing Eq.~\eqref{eq:app-SYK-kappa} to $p=i\mu$ and setting $\kappa(i\mu_B)=1$ gives
\begin{equation}
\cosh\mu_B=
\frac{1+\nu+(\eta-2)\nu^2}{\eta\nu^2},
\qquad u_B\equiv\frac1{\mu_B}.
\label{eq:app-SYK-muB}
\end{equation}
The energy-momentum tensor contribution has spatial dependence $e^{\tau-\mu_B|r|}$ only on rays beyond the saddle--pole crossing,
\begin{equation}
\frac{|r|}{\tau}>\left.\partial_\mu\kappa(i\mu)\right|_{\mu_B}
=\frac{\eta\nu^2\sinh\mu_B}{\nu+2}.
\label{eq:app-SYK-pole-domain}
\end{equation}
It cannot be extrapolated to $r=0$ at fixed $\nu<1$ to perform the spatial sum. At second order, that sum sets the two exchanged momenta to $p$ and $-p$. The strip-edge factor is then $2\cos[\pi\kappa(p)]$, and the small-momentum saddle gives
\begin{equation}
\int dp\,e^{2\kappa(p)\tau}\propto\frac{e^{2\nu\tau}}{\sqrt{\tau}}.
\end{equation}
The second-order commutator coefficient is proportional to $-\cos(\pi\nu)$. It is therefore positive for $\nu>1/2$ and negative for $\nu<1/2$. Comparing its magnitude with Eq.~\eqref{eq:SYK-early} gives, away from $\nu=1/2$,
\begin{equation}
\nu\tau_\times=\ln\frac{N}{q^2}+\frac12\ln\tau_\times+O(1).
\label{eq:app-SYK-tcross}
\end{equation}
Near strong coupling, evaluating the residue at zero momentum requires $(1-\nu)\tau\gg1$. For the second-order term to dominate while the expansion remains perturbative, we require
\begin{equation}
\cos^2\!\frac{\pi\nu}{2}\ll
\frac{(q^2/N)e^{\nu\tau}}{\sqrt{\tau}}
\ll\cos\!\frac{\pi\nu}{2},
\end{equation}
at fixed $\eta,\Delta$. This saddle estimate does not apply in the opposite limit, $(1-\nu) \tau \ll 1$.

Near strong coupling, this branch is the true OTOC butterfly front~\cite{Gu:2018jsv,ChoiMezeiSarosi2021}:
\begin{equation}
v_B=\frac{2\pi}{\beta}u_B.
\label{eq:app-SYK-vB}
\end{equation}
Equation~\eqref{eq:app-SYK-muB} also has a direct interpretation in terms of the spatial MSS bound of Section~\ref{sec:lyapunov-volume}.  In the leading eikonal profile, the energy-momentum tensor pole limits analytic momentum continuation to $|\mu|<\mu_B$. This does not establish the convergence domain or control the MSS normalisation errors of the exact correlator.  Moreover, along the energy-momentum tensor front, the deficit behaves as $e^{\tau-\mu_B|r|}$, so that $\lambda_F(v)=2\pi/\beta-\mu_B|v|$ there.  Thus, when its hypotheses hold, the bound in Eq.~\eqref{eq:Legendre-MSS-main} is saturated at the edge of the admissible domain, $\mu\to\mu_B$, while the pole-free mode obeys $\frac{2\pi}{\beta}\kappa(i\mu)<2\pi/\beta$ strictly inside it.

For $\epsilon_\nu=1-\nu\ll1$,
\begin{equation}
u_B=\sqrt{\frac{\eta}{6\epsilon_\nu}}[1+O(\epsilon_\nu)],
\label{eq:app-SYK-uB-strong}
\end{equation}
and, using Eq.~\eqref{eq:app-SYK-parameters},
\begin{equation}
v_B=\frac{\pi}{\sqrt3}\sqrt{\frac{\eta\mathcal J}{\beta}}
\left[1+O((\beta\mathcal J)^{-1})\right].
\label{eq:app-SYK-vB-strong}
\end{equation}

\subsection{Logarithmic scales  and the late-time status}

The nonlinear energy-momentum tensor eikonal contains a separate delay. Including smooth residue and external-vertex factors, define
\begin{equation}
\tau_B=\ln\frac{N}{q^2}+O(1),
\qquad t_B=\frac{\beta \tau_B}{2\pi}.
\label{eq:app-SYK-TB}
\end{equation}
The front position is then
\begin{equation}
|r|=u_B(\tau-\tau_B)+o(\tau)
\label{eq:app-SYK-front-position}
\end{equation}
when $\tau-\tau_B$ is  large and the front ray satisfies Eq.~\eqref{eq:app-SYK-pole-domain}. By contrast, the extrapolated pole-free scale obeys Eq.~\eqref{eq:app-SYK-Tpf}. At fixed $0<\nu<1$,
\begin{equation}
\tau_{\rm pf}=\frac1\nu\ln\frac{N}{q^2}+O(\ln\ln N),
\qquad
\tau_B=\ln\frac{N}{q^2}+O(1),
\end{equation}
but $\tau_B$ is the extrapolated intercept of the far-distance front. It therefore cannot be used to order the onset of local nonlinear growth relative to $\tau_\times$ and $\tau_{\rm pf}$.

Suppose now that the full strong-coupling eikonal is uniformly saturated to Eq.~\eqref{eq:app-SYK-Cinf} throughout the interior of Eq.~\eqref{eq:app-SYK-front-position}, up to an $o(\tau)$ transition region, and that the total exterior contribution is $o(\tau)$. Then
\begin{align}
\E_J\bar{C}^{(L)}(t)
&=2u_BC_\infty(\tau-\tau_B)+o(\tau)
\nonumber\\
&=4\left[\cos\!\left(\frac{\pi\nu}{2}\right)\right]^{4\Delta}
 v_B(t-t_B)+o(t),
\label{eq:app-SYK-late-Csum}
\end{align}
and the exact disorder identity gives
\begin{equation}
\E_JD^{(L)}(t)=\E_J\bar{C}^{(L)}(t).
\label{eq:app-SYK-late-D}
\end{equation}
For $t/t_B\to\infty$, this reduces to Eq.~\eqref{eq:SYK-late-main} at the level of the leading slope.

The strong-coupling limit must be taken sequentially. We keep $0<\epsilon_\nu\ll1$ fixed while taking the large-$N$, late-time eikonal limit, and only then examine $\epsilon_\nu\to0$.  Setting $\nu=1$ first is singular. In this limit, the pole and saddle approximations are nonuniform, and the fixed-$\Delta$ thermal plateau vanishes.  With the sequential convention, Eqs.~\eqref{eq:app-SYK-Cinf-strong} and \eqref{eq:app-SYK-vB-strong} give the leading slope
\begin{equation}
\E_JD^{(L)}(t)
=\E_J\bar{C}^{(L)}(t)
\sim
\frac{4\pi}{\sqrt3}
\sqrt{ \frac{\eta\mathcal J}{\beta}}
\left(\frac{\pi}{\beta\mathcal J}\right)^{4 \Delta} t.
\label{eq:app-SYK-late-strong}
\end{equation}

The remaining qualification is uniformity of the nonlinear expansion. Ref.~\cite{CHMS2023} shows that additional multi-string vertices are softened by transverse momenta near maximal chaos. At each fixed string order this supports the interior strong-coupling expansion, but the estimate needed above must be uniform in string number and must also control the full transition and exterior regions. We do not prove those bounds. Accordingly, Eqs.~\eqref{eq:SYK-late-main} and \eqref{eq:app-SYK-late-strong} are leading strong-coupling eikonal predictions. The exact equality $\E_J D=\E_J\bar C$ and the calculated thermal norms do not require this all-orders assumption. The early Lyapunov formula has the leading-large-$q$ and pre-scrambling qualifications stated above. The plateau additionally requires the weaker pointwise dynamical assumptions in Appendix~\ref{app:SYK-plateau}. Moreover, the integrated late-time coefficient also requires spatially uniform control.

\section{One-sided thermal norm and the unregularised dOTOC}
\label{app:MSS-temp}

This appendix is included for comparison with the original one-sided thermal dOTOC and is not used in the proof of the one-volume theorem. The MSS theorem bounds an exponential mode when an appropriately normalised analytic OTOC has such a mode in its domain of validity. It does not imply a positive Lyapunov exponent, nor does it guarantee that a chosen commutator has nonzero overlap with the fastest mode. Equation~\eqref{eq:app-SYK-edge-cos} makes the second point explicit: the symmetric commutator vanishes on a single mode with $\lambda=2\pi/\beta$, consistently with Ref.~\cite{BlakeLiu2021}.

The one-volume theorem in Section~\ref{sec:one-volume} is specific to the symmetrically regularised thermal norm. For comparison, define the one-sided norm
\begin{equation}
\norm{X}_{\rho,2}^2=\Tr(\rho X^\dagger X)
=\norm{X\rho^{1/2}}_{S_2}^2
\end{equation}
and the corresponding connected unregularised dOTOC. The block inequality of Lemma~\ref{lem:L4} does not directly apply to this norm because the thermal weight sits on one side of the product.

A general but non-optimal bound is still available under a stronger complex-time assumption. Since
\begin{equation}
a\rho^{1/2}=\rho^{1/2}\tau_{-i\beta/2}(a),
\end{equation}
one has
\begin{equation}
\norm{Xa}_{\rho,2}
\le\norm{\tau_{-i\beta/2}(a)}_\infty\norm{X}_{\rho,2}.
\label{eq:app-unreg-mult}
\end{equation}
If the continuation of each fixed local density to $-i\beta/2$ is uniform in volume, ordinary Gibbs clustering gives an $O(R_{\rm LR}(t)^{d/2})$ one-source norm, while ordinary clustering of distant source commutators makes their covariance tail summable. The sitewise near-source count then yields
\begin{equation}
D^{(L)}_{\rm unreg} (t) \le K_\beta(1 + |t|)^{2d}.
\label{eq:app-unreg-t2d}
\end{equation}
In $d=1$, the required fixed complex-time continuation is available for finite-range interactions~\cite{Araki1969,PGPH2023}; in higher dimensions it holds in the corresponding high-temperature regime, while a volume-uniform complex-time disk is not automatic in general~\cite{Bouch2015}.  At $\beta=0$, the regularised and unregularised definitions coincide, so the sharper $t^d$ theorem applies. We do not claim that Eq.~\eqref{eq:app-unreg-t2d} is optimal at finite temperature.

\bibliographystyle{apsrev4-2}
\bibliography{biblio}

\end{document}